\documentclass[10pt,twocolumn,twoside]{IEEEtran}
\usepackage{amsthm}

\usepackage{graphicx,cite,subfigure,bm}

\usepackage{mathtools}  
\usepackage{amsmath}
\usepackage{amssymb}
\usepackage{tabulary}
\usepackage{booktabs}

\usepackage{amsmath}

\usepackage{graphicx}
\usepackage{amsmath,amssymb}
\usepackage{algorithmic,algorithm}
\usepackage{color}
 \newtheorem{theorem}{Theorem}[section]
 
 \newtheorem{lemma}[theorem]{Lemma}
 \newtheorem{proposition}[theorem]{Proposition}

 \newtheorem{definition}[theorem]{Definition}
 \newtheorem{remark}[theorem]{Remark}
\usepackage{multicol}

\def\blue{\color{black}}

\def\blue{\color{black}}
\def\newblue{\color{black}}

\begin{document}

\title{\LARGE Optimal Pricing to Manage Electric Vehicles in \\Coupled Power and Transportation Networks }
 
\author{Mahnoosh Alizadeh,
Hoi-To Wai,
Mainak Chowdhury,
Andrea Goldsmith,
Anna Scaglione,
and Tara Javidi\thanks{This work was supported by the NSF CPS Grant 1330081 and by the U.S. DoE's Office of Electricity   through the Consortium for Electric Reliability
Technology Solutions (administered by LBNL).}
}

%
%\author{Mahnoosh Alizadeh\IEEEauthorrefmark{1},~Hoi-To Wai\IEEEauthorrefmark{2},~Mainak Chowdhury\IEEEauthorrefmark{1}, ~Andrea Goldsmith\IEEEauthorrefmark{1},~Anna Scaglione\IEEEauthorrefmark{2}, ~Tara Javidi\IEEEauthorrefmark{3}
%\thanks{
%\IEEEauthorrefmark{1} Stanford University
%\IEEEauthorrefmark{2} Arizona State University
%\IEEEauthorrefmark{3} University of California San Diego
%}}

  \maketitle

%%% ----------------------------------------------------------------------
\maketitle
%%% ----------------------------------------------------------------------
\begin{abstract}
We study the system-level effects of the introduction of large populations of Electric Vehicles on the power and transportation networks. We assume that each EV owner solves a decision problem to pick a cost-minimizing charge and travel plan. This individual decision takes into account traffic congestion in the transportation network, affecting travel times, as well as as congestion in the power grid, resulting in spatial variations in electricity prices for battery charging. We show that this decision problem is equivalent to finding the shortest path on an ``extended'' transportation graph, with virtual arcs that represent charging options. Using this extended graph, we study the collective effects of a large number of EV owners individually solving this path planning problem. We propose a scheme in which independent power and transportation system operators can collaborate to manage each network towards a socially optimum operating point while keeping the operational data of each system private. We further study the optimal reserve capacity requirements for pricing in the absence of such collaboration. We showcase numerically that a lack of attention to interdependencies between the two infrastructures can have adverse operational effects.
\end{abstract}

\section{Introduction: a tale of two networks}
Large-scale adaptation of Electric Vehicles (EV) will affect the operation of two cyber-physical networks: power and transportation systems \cite{5453787}. Each of these systems has   been the subject of decades of engineering research. However, in this work, we argue that the introduction of EVs will couple the operation of these two  critical infrastructures. We show that  ignoring this interconnection and assuming that the location of EV plug-in events follows an independent process that does not get affected by electricity prices can lead to instabilities in electricity pricing mechanisms, power delivery, and traffic distribution.  Hence,  we  propose control schemes that acknowledge this interconnection and move both infrastructures towards optimal and reliable operation.  

To achieve this goal, we show that an individual driver's joint charge and path decision problem  can be modeled as a  shortest path problem on an {\it extended transportation graph with virtual arcs}. We  use this extended graph  to study the collective result of all drivers making cost-minimizing charge and path decisions on power and transportation systems.  We then show that two non-profit entities referred to as the independent power system operator (IPSO) and independent transportation system operator (ITSO) can   collaborate   to find  jointly optimal electricity prices, charging station mark-ups, and road tolls, while keeping the data of each system private.  We show that this collaboration is necessary for correct price design. We further study the generation reserve requirements to operate the grid   in the absence of such collaboration.

\subsubsection*{Prior Art}
The study of mechanisms for coping with demand stochasticity and grid congestion  is at the core of power systems research. In particular,  EVs are acknowledged to be one of the primary focuses of   demand response (DR) programs. DR enables electricity demand to become a control asset for the IPSO.  For example, the authors in  \cite{5551257,5982648,contrt,xupan,tong12,evs2, tong,dlsp1,dlsp2,dlsp3,dlsp4,dlsp5,poola, 6622681,kefayati2010efficient,ma2013decentralized,tushar2012economics}, and many others, have  proposed control schemes to manipulate EV charging load using various tools, e.g., heuristic or optimal control, and towards different objectives, e.g., ancillary service provision, peak shaving, load following. However,  a common feature in  \cite{5551257,5982648,contrt,xupan,tong12,evs2, tong,dlsp1,dlsp2,dlsp3,dlsp4,dlsp5,poola, 6622681,kefayati2010efficient,ma2013decentralized,tushar2012economics} is that the location and time of plug in for each request is considered an {\it exogenous} process and is not explicitly modeled. Very few works have considered the fact that, unlike all other electric loads, EVs are mobile, and hence may choose to receive charge at different nodes of the grid following economic preferences and travel constraints.  This capability was considered in \cite{6547832,bayram2012smart}  in the problem of routing EV drivers to the optimal nearby charging station after they announce their need to charge. The authors in \cite{khodayar} consider the case where the operator tracks the mobility of large fleets of EVs and their energy consumption and designs optimal multi-period Vehicle-to-Grid strategies. Here we do not consider the case of fleets and look at a large population of heterogeneous privately-owned  EVs.

\begin{figure}[t]
\centering
\includegraphics[width = \linewidth]{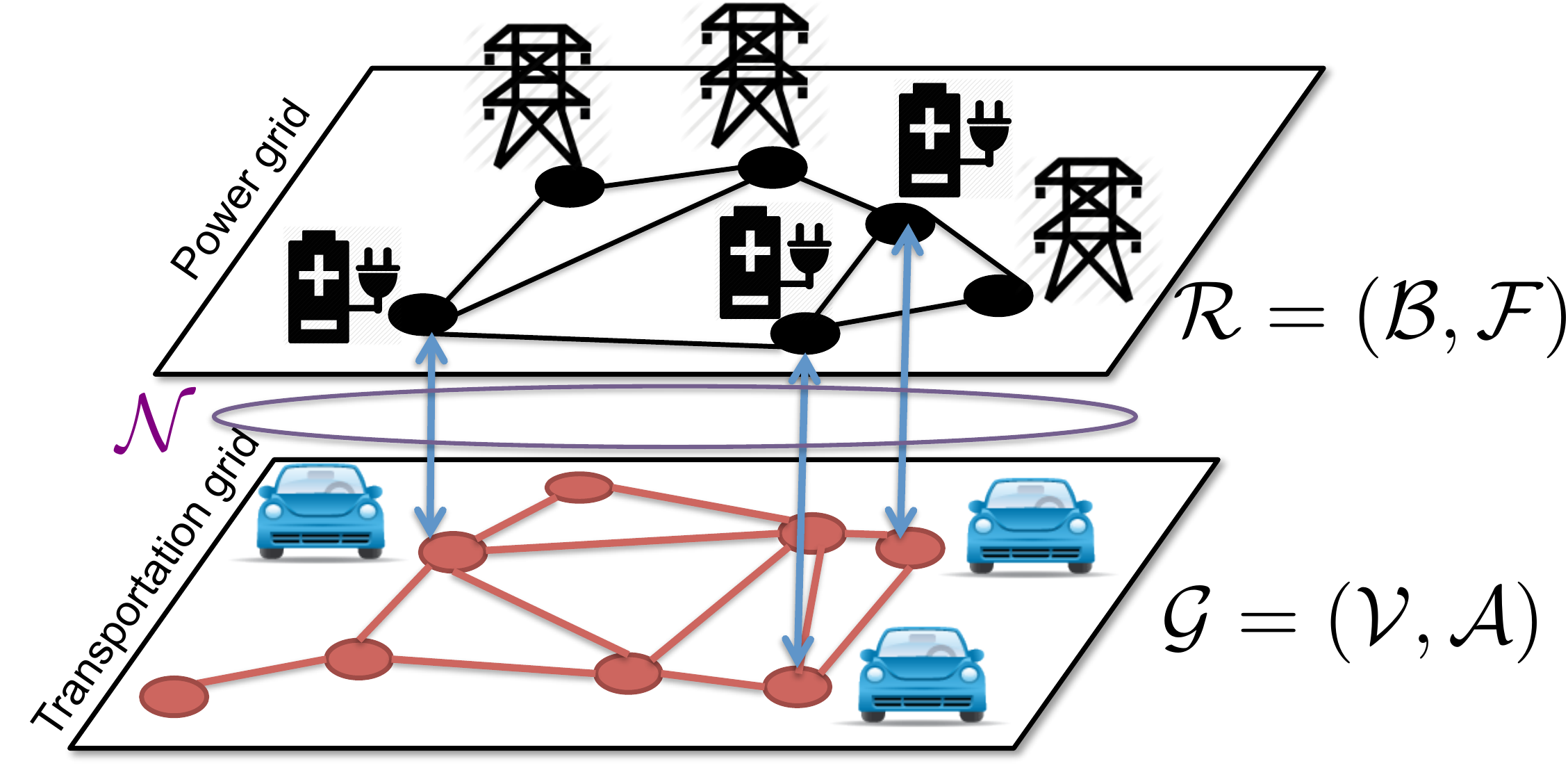}
\caption{Electric Vehicles affect transportation and power delivery networks.} 
\label{transnet}
\end{figure}

Traffic engineering studies mechanisms 
 for coping with road congestion  in the transportation network.
At the individual user level, travel paths are planned to avoid congestion as much as possible, naturally leading to shortest path problems \cite{dreyfus1969appraisal, chabini1998discrete}. When studying the collective actions of users, the so-called Traffic Assignment Problem is concerned with the effect of individuals' selection of routes on the society's welfare, and studies  control strategies to guide the selfish user equilibrium towards a  social optimum, e.g., \cite{peeta2001foundations}, \cite{wardrop1952road}. 

Recently, a line of  research has emerged to study the effect of EVs on transportation systems. For example,  \cite{sachenbacher2011efficient, fontana2013optimal, artmeier2010shortest} look at the individual path planning problem by minimizing the energy consumption of  EVs, leading to a constrained shortest path
problem. However, the interactions with the power grid are not modeled. 
{\blue At the system-level,  \cite{he2014network, wang2014energy} study efficient solutions for characterizing the redistribution of traffic due to the charge requirements of EVs  (paths are forbidden if not enough charge is received to travel them). In contrast to our work,   in \cite{he2014network, wang2014energy}, electricity prices are respectively not considered and taken as given. Accordingly, these works are complementary to ours and do not address the electricity price design aspect that we are interested in (more details in Remark \ref{remarkdp}). To the best of our knowledge, the only work that considers  price design  is \cite{He20131}. In \cite{He20131},  charge is wirelessly delivered to EVs while traveling. Hence, EVs can never run out of charge. The authors show that  if a government agency controls the operations of both the transportation and power networks or can design tolls as a decentralized control measure, the effect of EVs on the  grid can be optimized by affecting the drivers' choice of route. In spite of a somewhat similar set up, \cite{He20131} and our work have major differences: 1) Our model is different in that we assume EVs make stops at charging stations and the amount of charge received is a choice made by the driver, leading to a different pricing structure, based on the concept of virtual charging arcs; 2) We consider  the IPSO and ITSO as two separate entities and look at how  they can design prices if they collaborate together with minimal data exchange using the principles of dual decomposition. We  also study the adverse effects of the lack of such collaboration; 3) We study how the IPSO can set   prices in the absence of such collaboration through procuring generation reserves.  }

\vspace{-0.2cm}

\begin{remark}
To be able to derive analytical results, we have chosen to remain in a static setting. This means that the customers' travel demand, the baseload, and generation costs are all time-invariant. Our preliminary work published as a conference paper \cite{allerton2014}  models this problem under a dynamic setting.  {\blue The main contribution of  [32] is proposing the general model of the ESPP and the extended transportation graph.
However, the dynamic model studied in [32] in its full generality was not amenable to an analytical characterization of the aggregate control problem and hence could not provide design insights. In contrast, the present work introduces significant simplifications by considering a static setting. The static formulation  removes the non-convexity of the problem and allows for a novel analytical treatment.}
\end{remark}

\section{Overview}
{\blue We study a large network of EV and Internal Combustion Engine Vehicles (ICEV) owners that optimize their daily trip costs, including the path they take to complete a trip as well as refueling strategies. A short model of the decision making process by individual EV drivers is first presented in Section III, mainly to introduce the {\it extended transportation graph}, a novel concept we use in this paper to integrate individual decisions into system-level control strategies  for coupled infrastructures.
The extended graph construct captures the fact that EVs' route and charge decisions are affected by the state of two networks, namely the power and the transportation networks.  The transportation network is managed by a non-profit ITSO (red circle in Fig. 2), who knows about the trip patterns of the  population and can impose tolls on public roads to affect the individuals' routing decisions. The power network is managed by a non-profit IPSO (light gray circle), who controls electricity generation costs  (green circle) and is in charge of pricing electricity that affect individual EV's charging decisions. Ideally, we would like   to minimize the total transportation delay and electricity generation cost that the society incurs. However, as the IPSO and the ITSO are two separate entities, we study whether they can achieve this goal with or without direct collaboration under various schemes  presented in Section \ref{sec.so} (and in Fig \ref{schemes}). We numerically study these schemes in Section \ref{sec.numerical}.}

\section{The Individual Driver's Model}\label{sec.ind}
Let us first focus on the decision making process of an individual EV driver (the blue circle in Fig 2).  In order to complete a trip, the driver needs to decide on 1) which path to take to get from his origin to the destination; and 2)  the locations at which he/she should charge the EV battery and the amount of charge to be received at each location.  We model the cost structure associated with these decisions next.

\begin{figure}[t]
\centering
\includegraphics[width = 0.75 \linewidth]{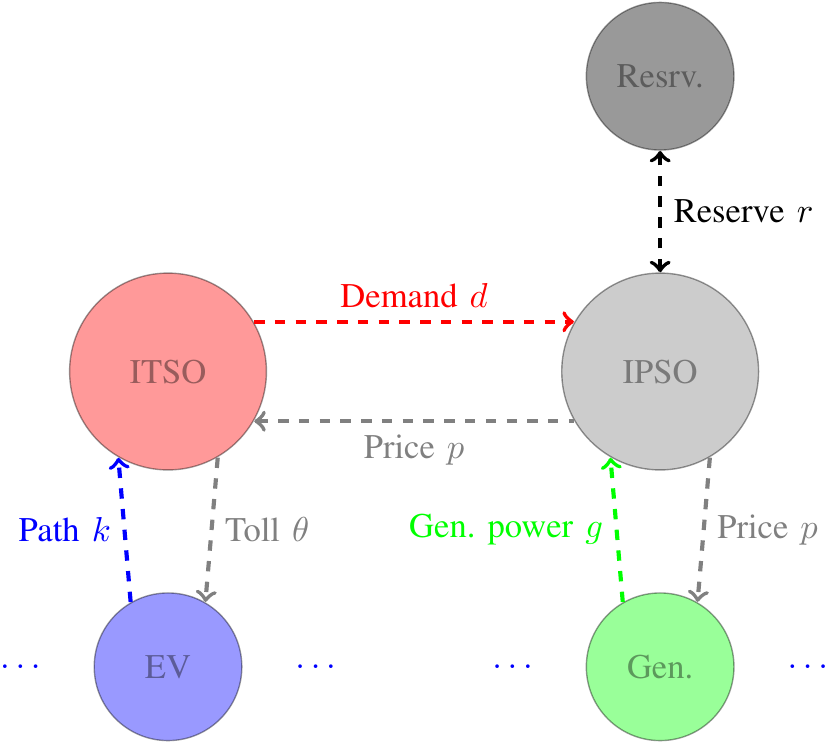}
\caption{The entities involved in the control problem.} 
\label{transnet}\vspace{-0.4cm}
\end{figure}

\subsubsection*{Notation}  We use bold lower case $\mathbf x$ to indicate vectors and bold upper case $\mathbf X$ to indicates matrices. The notation ${ \mathbf x_{\cal I}} = [x_i]_{i \in \mathcal I}$ indicates that the elements that comprise a column vector or a matrix each correspond to a member of a set $\mathcal I$. The symbols $\preceq$ and $\succeq$ denote element-wise  $\leq$ and $\geq$ inequalities in vectors. The transpose of a column vector $\mathbf x$ is denoted by $\mathbf x^T$. The all one and all zero row vectors of size $j$ are denoted by $\mathbf{1}_{1\times j}$ and $\mathbf{0}_{1\times j}$ respectively.

\begin{table}[htbp]\caption{\blue Table of Notation }
\begin{center}% used the environment to augment the vertical space
% between the caption and the table
\begin{tabular}{r c p{0.6\linewidth} }
\toprule
$\mathcal V$ & $\triangleq$ & Set of nodes in the transportation network\\
$\mathcal A$ & $\triangleq$ & Set of arcs in the transportation network\\
$\mathcal N$ & $\triangleq$ & Set of nodes with charging facilities\\
$G$ & $\triangleq$ & The transportation graph\\
$\mathcal K$ & $\triangleq$ &  Set of energy-feasible paths that connect the origin and destination for an individual user\\
$s_k$ & $\triangleq$ &  Length of path $k$\\
$\lambda_a$ & $\triangleq$ &Flow on arc $a$ of transportation graph\\  
$\lambda_v$ & $\triangleq$ & Flow into charging station located at node $v$\\
$p_v$ & $=$ & Price of electricity at node $v$\\
$e_v$ & $=$ & Energy received at node $v$\\
$\mathcal E_v$ & $\triangleq$ & Set of possible charging amounts at $v$ \\
$\theta_v$ & $=$& Plug-in fee at node $v$\\
$\rho_v$  &$\triangleq$ & Charging rate at node $v$\\
$\tau_a(\lambda_a)$ & $\triangleq$ & Latency function of traveling on arc $a$\\
$\tau_v(\lambda_v)$ & $\triangleq$ &Wait time to be plugged in at node $v$\\
$\gamma$ &$\triangleq$ & Value of time to users\\
$s_a(\lambda_a)$ & $\triangleq$ & Cost each user incurs for traveling on arc $a$\\
$s_ v (e_v, \lambda_v)$ & $\triangleq$ & Cost  to receive charge of $e_v$ at node $v$\\
$e_a$ & $\triangleq$ & Energy required to travel arc $a$\\
$G^e$ &$\triangleq$&  The extended transportation graph \\
$\mathcal S$ & $\triangleq$ & Set of nodes in $G^e$\\
$\mathcal L$ & $\triangleq$ & Set of arcs in $G^e$\\
$\mathcal L_v$ & $\triangleq$ & Set of virtual charging arcs for charging station at node $v$\\
$\mathcal C$ & $\triangleq$ & Set of all virtual charging station entrance  and bypass arcs\\
$b_a$ & $\triangleq$ & The electricity bill to charge for virtual arc $a$ \\
$\mathcal Q$ & $\triangleq$ & Set of different origin-destination clusters $q$\\
$\mathcal K_q$& $\triangleq$ &Set of feasible paths on $\mathcal G^e$ for cluster $q$ \\
$m_{q}$ & $\triangleq$ & Rate of EVs in cluster $q$ \\
$f_{q}^k$ &$=$& rate of  cluster $q$ EVs that choose path $k \in \mathcal{K}_{q}$\\
$\mathbf{f}_q$ & $\triangleq$ & $[f_q^k]_{k \in \mathcal{K}_q}$\\
$\mathbf{A}_q$ & $\triangleq$ &Arc-path incidence matrix for cluster $q$\\
$\mathbf g$ & $=$ & Vector of generation outputs at all network nodes\\
$\boldsymbol{c} (\mathbf{g})$& $\triangleq$ & Vector of network generation costs \\
$\mathbf u$ & $\triangleq$ & Vector of inelastic non-EV demand at all network nodes\\
$\mathbf d$ & $=$ & Vector of EV charging demand at all network nodes\\
$\mathbf H$ & $\triangleq$ & The power transfer distribution matrix\\
$\mathbf c$ & $\triangleq$ & Line flow limits\\
$\mathbf M$ & $\triangleq$ & Matrix that maps virtual link flow to power system load\\
$s^e_a(\lambda_a)$& $=$ & Auxilary cost function for arc $a$ (see \eqref{aux})\\
$\boldsymbol{w}^e(\boldsymbol{\lambda}) $& $=$ & $ [\int_{x=0}^{\lambda_a}s^e_a(x) dx]_{a\in \mathcal L}$\\
$\boldsymbol \xi$  & $\triangleq$ & Reserve capacity prices at all network nodes\\
$\boldsymbol r$  & $\triangleq$ & Reserve capacity prices at all network nodes\\
\bottomrule
\end{tabular}
\end{center}
\label{tab:TableOfNotationForMyResearch}
\end{table}

\vspace{-0.3cm}

\subsection{Congestion Costs}\label{transcost}
 We model the transportation network through a  connected directional graph $\mathcal{G} = (\mathcal{V},\mathcal{A})$, where $\mathcal{V}$ and $\mathcal{A}$  respectively denote the set of 
 nodes and arcs of the graph.

Traveling on the transportation network is associated with a cost for the user since he/she values the time spent en route. The time to travel between an origin and a destination node  is comprised of the time spent on arcs that connect these two nodes on $\mathcal{G}$.  Here we adopt the popular Beckmann model for the cost of traveling an arc, i.e., a road section \cite{beckmann1956studies}. Accordingly, we assume that the travel time for each arc $a\in \mathcal{A}$ only depends on
 the rate of EVs per time unit that travel on the arc, which we refer to as the arc  flow and denote by $\boldsymbol{\lambda} = [\lambda_a]_{a\in \mathcal{A}}$. The time it takes to travel arc $a$ is then represented by a latency function $\tau_a(\lambda_a)$, which is convex, continuous, non-negative, and increasing in $\boldsymbol{\lambda}$. The congestion  cost that a user incurs for traveling arc $a$ is given by:
\begin{equation}\label{incon}s_a(\lambda_a)= \gamma \tau_a(\lambda_a),\end{equation}
where $\gamma$ is the cost of one unit of time spent en route. Hence, the cost to travel link $a$ is:
\begin{equation} \mbox{total cost to travel link $a$} = s_a(\lambda_a) + \theta_a, \end{equation}
where $\theta_a$ corresponds to any tolls that the driver should pay to the ITSO, if any such toll is enforced for link $a$.

Moreover, traveling each arc $a\in \mathcal{A}$ requires a certain amount of energy $e_a$ (see Fig. 2). Energy needs to be received  from the power grid and stored in the EV battery. The cost that the user incurs to receive battery charge is modeled next.

\vspace{-0.3cm}

\begin{figure}
\centering
\includegraphics[width= 0.7\linewidth]{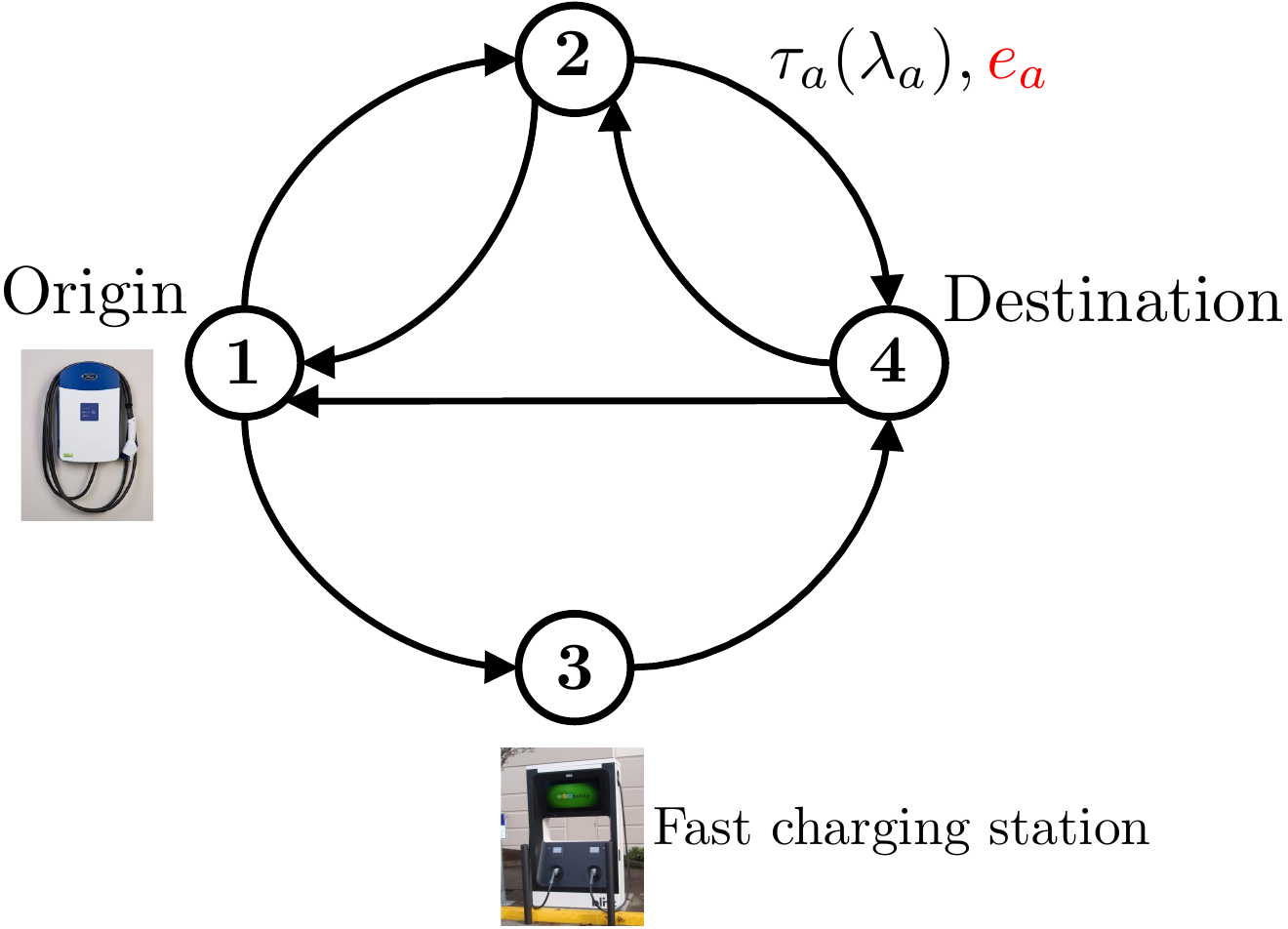}\label{real}
  \caption{ The transportation graph $\mathcal G$.}
\vspace{-0.3cm}
\end{figure}

\subsection{Charging Costs}
A subset of nodes  on the transportation network $\mathcal{N} \subseteq \mathcal{V}$ are equipped with battery charging facilities and, hence, the EV drivers have the choice of charging their batteries at these locations. Naturally, to be able to provide charging services, the  nodes $\mathcal{N}$ are also a subset of the nodes $\mathcal{B}$ that constitute the power grid graph $\mathcal{R} = (\mathcal{B}, \mathcal{F})$. Each node $v \in \mathcal{B}$ has an associated price for electricity $p_v$. Consequently, if the EV driver chooses to charge at location $v$, he/she will pay:
\begin{equation} \mbox{electricity cost of charging}= p_v e_v + \theta_v, \label{bill}\end{equation}
where $e_v \in \mathcal{E}_v$ is the battery charge amount received at $v$ chosen from a finite set $\mathcal{E}_v$, and $\theta_v$ corresponds to a one-time plug-in fee for the charging station at $v$. Moreover, if $v$ is not the origin  of the trip, the driver needs to spend some extra time en route in order to receive charge at $v$. This is due to any congestion at the fast charging stations (FCS) plus the time it takes to receive charge. Hence, an extra inconvenience cost is incurred. If the charging rate at FCS $v$ is $\rho_v$, and the rate of EVs  being plugged in for charge at this location is denoted by $\lambda_v$, this cost is equal to:
\begin{equation}s_ v (e_v, \lambda_v)\!\!= \!\!\begin{cases} \gamma\left(\frac{e_v}{ \rho_v} + \tau_v(\lambda_v)\right),~~~ \mbox{$v$ is an FCS}\\ 0,~~~~~~~~~~~~~~~~~~ \mbox{$v$ is trip origin node} \end{cases}\label{incon.charge}\end{equation}
where $\tau_v(\lambda_v)$ denotes the wait time due to congestion, and is a soft cost model to capture limited station capacity.

\subsection{The Extended Transportation Graph with Virtual Arcs}
The EV driver's goal would be to find the least cost path that connects the origin and the destination on the transportation graph $\mathcal{G}$ while making sure that the EV battery never runs out of charge\footnote{\blue Note that with the cost structure we have defined, the EV driver will reach the destination with minimum-possible leftover charge. An extension of the analytical framework to include the value of leftover charge at the destination in the optimization is trivial and has been removed for brevity of notation. We refer the reader to our conference paper  \cite{allerton2014}  for this extension.}.

Here we show that we can recast the EV driver's route and charge problem as a resource-constrained shortest path problem  on a new {\it extended transportation {multi}graph} $G^e(\mathcal{S}, \mathcal{L})$. This definition will help us study and control the aggregate effect of individual EVs on power and transportation systems.
 
\begin{definition}   A  travel {\it path} $k$  on the graph $\mathcal{G}$
is characterized by an ordered sequence of arc indices  $\mathbf{a}_k$ where
the head node of $[\mathbf{a}_k]_i$ is identical to the tail node of  $[\mathbf{a}_k]_{i+1}$ for all $i= 1, \ldots, s_k-1$. The length of  path $k$ is $s_k$. Alternatively, if $\mathcal{G}$ is simple, path $k$ can be written as an ordered sequence of $s_k$ node indices $\mathbf{v}_k$ (excluding the destination node).   We further denote a vector of previously defined quantities associated with the arcs $\mathbf{a}_k$  using subscripts, e.g.,   $\boldsymbol{e}_{\mathbf{a}_k} = [e_a]_{a \in \mathbf{a}_k}$ and $\boldsymbol{s}_{\mathbf{a}_k}(\boldsymbol{\lambda}_{\mathbf{a}_k})$ respectively denote the vectors of charge amounts and travel time required to travel each arc on path $k$.
\end{definition}

%
%Accordingly, we define a vector of path-specific arc flow, vector of energy required to travel arcs on path, and vector of cost functions for path $k$, respectively, as:
%$$\boldsymbol{\lambda}_{\mathbf{a}_k} = [\lambda_a]_{a\in  \mathbf{a}_k},~ \boldsymbol{e}_{\mathbf{a}_k} = [e_a]_{a \in \mathbf{a}_k},~ \boldsymbol{s}_{\mathbf{a}_k}(\boldsymbol{\lambda}_{\mathbf{a}_k})  = [s_a(\lambda_a)]_{a \in \mathbf{a}_k}$$
%and vectors of nodal charge energy and nodal delay costs as:
%$$\boldsymbol{e}_{\mathbf{v}_k} = [e_v]_{v \in \mathbf{v}_k}, ~ \boldsymbol{s}_{\mathbf{v}_k}(\boldsymbol{e}_{\mathbf{v}_k} ) = [s_v(e_v)]_{v \in \mathbf{v}_k}.$$ 

To complete a trip, the driver incurs two forms of costs: the cost associated with arcs and the cost associated with the charging decisions taken at nodes.  However, we can observe charging is very similar in nature to traveling: 1) it takes a certain amount of time (due to both charging rate limitations  and congestion); 2) it has a cost; and 3) it changes the energy level of the battery.
Acknowledging this similarity, we  transform the EV driver's decision problem to a shortest path problem by  associating charging decisions made at the nodes $v$ of the transportation graph to a set of new {\it virtual arcs} to be added at each node $v \in \mathcal{N}$ where charging is possible. At each origin and destination node where charging is possible, the following transformation would capture all decisions:
\begin{itemize}
\item The decision of how much to charge: a set of virtual arcs $\mathcal{L}_v$ added at node $v$ are each associated with a specific choice of how much to charge, i.e., $e_v \in \mathcal{E}_v$.  Hence, the energy {\it gained} by traveling each new virtual arc is set to be one such  member of $\mathcal{E}_v$ (red arcs in Fig. 4). Equivalently, we can say that the energy required to travel the virtual arc is negative. Travel time is  $\frac{e_v}{ \rho_v}$.
\end{itemize} 
At the FCS, these transformations  capture all decisions:
\begin{itemize}
\item The decision to stop or skip stopping at a charge station en route:   the driver can either take a {\it charging station entrance arc} (labeled {\it ``stop''} in Fig 3) and plug in their EV at the station, or skip stopping at the station via a {\it bypass arc} with zero travel time and energy requirements (green arcs in Fig. 3). Charging station entrance arcs can be congested;
\item The decision of how much to charge if stopping at $v$: the charging station entrance arc is connected in series to a set of virtual arcs $\mathcal{L}_v$  capturing the choice of amount of charge $e_v \in \mathcal{E}_v$ (blue arcs in Fig. 3).
\end{itemize}
The flow on the charging station   entrance arc will capture the wait time to be plugged in at the station.    The set of all entrance and bypass arcs for all charging stations is $\mathcal C$. The new extended transportation graph with these virtual arcs would then have the following set of arcs:
$$\mathcal{L} = (\cup_{v\in \mathcal{N}}  \mathcal{L}_v)  \cup \mathcal{C} \cup \mathcal{A}.$$

\begin{figure}
\includegraphics[width=\linewidth]{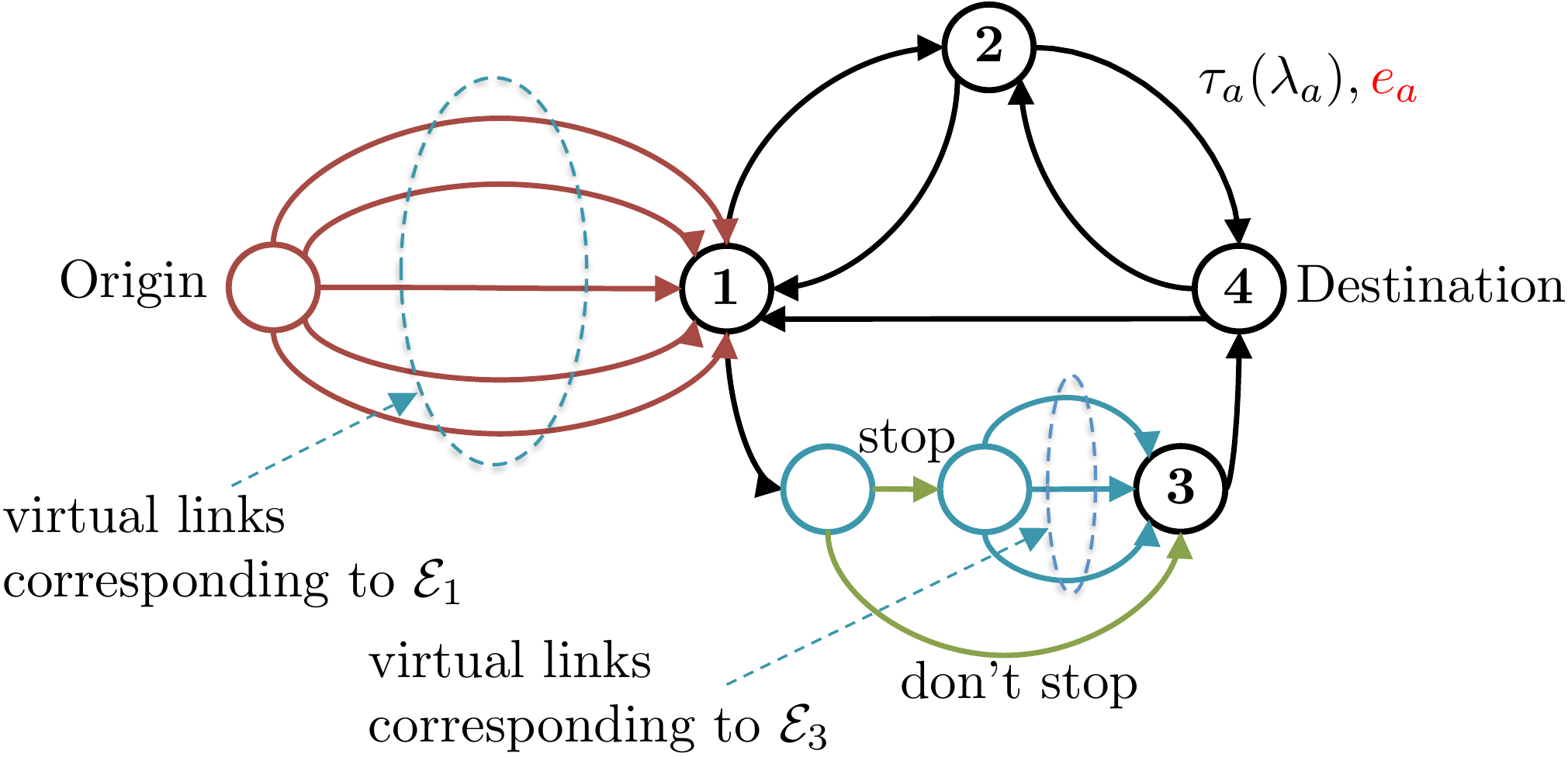}
  \caption{The extended graph corresponding to Fig. 2.}
\vspace{-0.3cm}
\end{figure}

Consequently, the transformed problem seeks
a shortest path on this extended graph from the origin to the destination, with the cost of traveling each arc being the sum of its travel time cost   and all monetary charges such as the electricity bill or tolls. The travel time costs on $G^e(\mathcal{S}, \mathcal{L})$ is given by:
\begin{align} \label{costslambda}
s_a(\lambda_a) &= \begin{cases}  \gamma \tau_a(\lambda_a),~~~~~~~~~~~~~~~~~~~~~~a\in \mathcal{A} \cup  \mathcal{C}\\
\gamma \frac{e_a}{\rho_v},~~~~~~~~~~~~a \in \mathcal{L}_v, v\in \mathcal{N} - \mbox{origin}\\
0,~~~~~~~~~~~~~~~~~~~~~~a \in \mathcal{L}_v, v = \mbox{origin}
\end{cases}
\end{align}
All other monetary costs can be captured as:
\begin{align}
b_a &= \begin{cases}\theta_a,~~~~~~~~~~~~~~~a\in \mathcal{A}\cup  \mathcal{C}\\
p_v e_a,~~~~~~~a \in \mathcal{L}_v, v\in \mathcal{N} 
\end{cases}
\end{align}
and hence, each driver selfishly optimizes  their trip plan by solving the Energy-aware Shortest Path Problem (ESPP):
\begin{align} \label{ind.opt2}
&\min_{k \in \mathcal{K}} \mathbf{1}_{1\times s_k} ( \boldsymbol{s}_{\mathbf{a}_k}(\boldsymbol{\lambda}_{\mathbf{a}_k}) +  \boldsymbol{b}_{\mathbf{a}_k}) 
\end{align}
where $\mathcal{K}$ is now the set of energy-feasible paths that connect the origin and the destination on the extended  graph $G^e(\mathcal{S}, \mathcal{L})$. Energy-feasibility of a path ensures that the battery will never run out of charge en route. Next, we define energy-feasibility mathematically.

{\blue 
\begin{definition}
 A path $k$ is energy-feasible on $G^e$ iff $\forall j = 1,\ldots, s_k$:
\begin{align}
 &0 \leq \mbox{Initial charge} - [\mathbf{1}_{1\times j}, \mathbf{0}_{1\times ({s_k - j})}] \boldsymbol{e}_{\mathbf{a}_k}  \leq \mbox{battery capacity}. \nonumber
\end{align}
\end{definition}
These two constraints ensure that the EV never runs out of charge en route  if taking path $k$, and that the battery charge state  never exceeds battery capacity. A vector of dimension zero is simply empty.
Note that with this definition, we can determine whether a path is energy-feasible  independently of the network congestion mirrored through $\boldsymbol{\lambda}_{\mathbf{a}_k}$.  Hence, for system-level control of $\boldsymbol{\lambda}_{\mathbf{a}_k}$, the set of energy-feasible paths can be calculated offline.}
 
\begin{remark}\label{remarkdp}
The ESPP \eqref{ind.opt2}  can be solved using Dynamic Programming (DP) algorithms with pseudo-polynomial complexity  \cite{cai1997time}.  {\blue Polynomial-time Dijkstra-like algorithms for  solving the shortest path problem cannot be applied due to the existence of the energy-feasibility constraint (see \cite{sachenbacher2011efficient, dror1994note}). This is mainly because the cost of a path is no longer just
the sum of its  arc costs (as energy constraints cannot be
attributed to individual arcs but a sum over multiple arcs).
Proposing efficient solution methods for the ESPP is beyond the scope of this paper. Instead, our focus is to use the extended graph to study aggregate control strategies. We refer the reader to recent papers   studying efficient solutions and search heuristics for variations of ESPP, e.g., \cite{wang2014energy,timewindowsev,6183286,storandt2012quick}. For our small numerical experiment, we use a brute force approach to enumerate all loop-free energy-feasible paths for all origin-destination pairs on the extended transportation graph, as often done in the transportation literature. While for our small  experiment the complexity of this approach does not pose a computational challenge, in more realistic models this is an issue that needs to be properly addressed to allow scalability. We will consider this issue as part of our future work.     }
\end{remark}

\vspace{-0.1cm}

Now, imagine that    every EV owner in the society solved \eqref{ind.opt2} to plan their trips. These users share two  infrastructures: the transportation network, and the power grid. Hence, collectively, EVs give rise to a traffic and load pattern, determining which roads and grid buses are congested and hence, will  have longer travel times and higher electricity prices. Through this interaction, individuals affect each other's cost, leading to the system-level problem that we are interested in.

\section{System Level Model } \label{sec.so}

At the system level, the extended transportation graph helps us  study the collective effect of individual drivers  on  traffic and energy loads  as a network flow problem.    Here virtual arcs are added   at all potential origins and FCSs.

At the aggregate level, the system variables, i.e., the flow rate of vehicles on arcs $\lambda_a$, the  price of electricity $p_v$, and the  tolls $\theta_a$ can no longer be considered 
  as variables imposed upon the system but rather as variables to be jointly optimized by system operators.

If a single  entity was in charge of monitoring the state of both networks and controlling all EVs' charge and route decisions, they can maximize the social welfare by solving for the optimum route and charge plan for each individual such the the total transportation congestion and generation costs that the society incurs is minimized. In doing so, this entity needs to ensure that the constraints of the  transportation and power systems are not violated.
The power system's constraints ensure the balance of supply and demand in the grid and that the physical limits of transmission lines are not violated. The transportation system constraints ensure that every driver will be able to finish their travel. 

In reality, power and transportation systems are operated by the IPSO and the ITSO respectively, and their operational data is not shared. The IPSO is in charge of optimizing generation costs subject to power system constraints, and the ITSO is in charge of optimizing transportation costs subject to transportation system constraints. Also, individuals' route and charge decisions can only be affected through prices. We first study the ITSO and IPSO   strategies separately in Subsections \ref{sec.itso.problem} and \ref{sec.ipso.problem} and their optimal pricing. Then we look at how they  interact  (and possibly  achieve the socially optimal outcome) in Section V.

\vspace{-0.2cm}

\subsection{The ITSO's Charge and Traffic Assignment Problem}\label{sec.itso.problem}
 {\blue We assume that drivers  belong to   a finite number of classes $q \in \mathcal Q$. Vehicles in the same class share the same origin and destination. Vehicles could include both EVs as well as ICEVs.  A given class q would contain either EVs or ICEVs but not both.  Drivers of the same class $q$ are represented by a  set of  feasible paths $\mathcal K_q$, each of which allows them to to finish their trip on the extended graph. For EVs, this is equivalent to the set of energy-feasible paths given in Definition  II.2 and can be enumerated offline for each class. For ICEVs, we can assume these paths simply include transportation arcs in $\mathcal{A}$ that connect the origin and destination, and do not enter charging stations.  Clearly, any other path selection method that considers more realistic constraints can also be applied.} We leave the study of optimal clustering mechanisms that assign heterogeneous users to a finite of number of classes to future work.

Each customer  directly affects the flow  of the arcs that constitute his/her path. To model this effect, we define:
\begin{itemize}
\item $m_{q}$ as the travel demand rate (flow) of EVs in cluster $q$. This demand is taken as deterministic  and given;
\item $f_{q}^k$ as the rate of  cluster $q$ EVs that choose path $k \in \mathcal{K}_{q}$. We define $\mathbf{f}_q = [f_q^k]_{k \in \mathcal{K}_q}$.
\end{itemize}
 Naturally, since every driver has to pick one path, the following conservation rule holds:
\begin{equation} \mathbf{1}^T \mathbf{f}_q =m_{q}.\end{equation}

Given the   path decisions of all EV drivers, i.e., the $f_{q}^k$'s, the flow of EVs on arc $a$ is given by $\lambda_a = \sum_{q\in \mathcal{Q}, k \in \mathcal{K}_{q}}   \delta_{a}^k  f_{q}^k$,
where $\delta_{a}^k$ is an arc-path incidence indicator (1 if arc $a$ is on path $k$ and 0 otherwise). This  is written in matrix form as:
\begin{equation}\label{ftolambda2} \textstyle
\boldsymbol \lambda = \sum_{q\in \mathcal{Q}}  \mathbf{A}_q  \mathbf{f}_q, 
\end{equation}
where $\boldsymbol{\lambda} = [\lambda_a]_{ a\in{\mathcal L}}$ denotes the vector of network flows and $\mathbf{A}_q$ is a $|\mathcal L| \times |\mathcal K_q|$ matrix such that  $[\mathbf{A}_q]_{a,k} =   \delta_{a}^k$.

The flow on the virtual arcs of the  extended graph  leads to a power load.
We denote the charging demand at each node $v \in \mathcal B$ of the grid as a vector $\mathbf d = [d_v]_{v \in \mathcal{B}}$, given by:
\begin{equation}\mathbf{d} = \mathbf M  \boldsymbol{\lambda}, \label{lambdatod}\end{equation}
where  $\mathbf M$ is a $|\mathcal B| \times |\mathcal L|$ matrix given by:
\begin{equation}[\mathbf M]_{v,a} = \begin{cases} e_a,~~a\in \mathcal{L}_v, v \in \mathcal{N} \\0,~~~~\mbox{else}\end{cases} \label{mdef}\end{equation}

Let $\boldsymbol{s}(\boldsymbol{\lambda}) = [s_a(\lambda_a)]_{a\in \mathcal L}$.
If an ITSO is in charge of determining the optimal path and charge schedule for each EV  such that the aggregate   cost is minimized, it  can solve a modified version of the classic static traffic assignment problem  \cite{dafermos1969traffic} on the extended graph, which we refer to as the charge and traffic assignment problem (CTAP): 
\begin{align} \label{soc.itso}
&\min_{\mathbf{f}_q, q \in {\cal Q} } \boldsymbol{\lambda}^T \boldsymbol{s}(\boldsymbol{\lambda}) + \mathbf p^T  \mathbf d \\& \mbox{s.t.}~~ (\star)  \begin{cases} \mathbf{f}_q  \succeq \mathbf 0,~\forall q\in \mathcal{Q}, \\ \mathbf{1}^T \mathbf{f}_q = m_{q},~\forall q\in \mathcal{Q},\\ \boldsymbol \lambda = \sum_{q\in \mathcal{Q}}  \mathbf{A}_q  \mathbf{f}_q,~\mathbf{d} = \mathbf M  \boldsymbol{\lambda},\end{cases}\nonumber
\end{align}
where $\mathbf p = [p_v]_{v \in \mathcal B}$.

\vspace{-0.2cm}
\subsection{The IPSO's Economic Dispatch Problem}\label{sec.ipso.problem}

To serve the charging demand of EVs, a set of generators are located at different nodes of the power network $\mathcal{R} = (\mathcal{B}, \mathcal{F})$.  For brevity, let us assume that a single merged generator is located at each node of the grid. Assuming that the generation at each node is denoted by a vector $\mathbf g = [g_v]_{v \in \mathcal{B}}$ and the baseload (any load that does not serve EVs) by a vector $\mathbf u = [u_v]_{v \in \mathcal{B}}$, there are three constraints that define a {\it feasible generation mix} $\mathbf g$ in the power grid. First of all,  $g_v$ must be within the capacity range of the generator at node $v$, i.e., $\mathbf g^{\min} \preceq \mathbf g \preceq \mathbf g^{\max}$.
Second, the demand/supply balance requirement of the power grid should be met, i.e.,
\begin{align} &  \mathbf{1}^T  ( \mathbf d  + \mathbf u - \mathbf g)     =  0. \end{align}  
Third and last, the transmission line flow constraints of the grid under the  DC approximation \cite{verma2009power} translate into:
\begin{align} & \mathbf H ( \mathbf d + \mathbf u - \mathbf g)  \preceq \mathbf c,  \end{align}  
where the matrix $\mathbf H$ is the power transfer distribution matrix of the  grid, explicitly defined in \cite{verma2009power}, and $\mathbf c = [c_f]_{f\in \mathcal F}$ is a vector containing arc (line) flow limits (in both directions).

In most power grids, one such feasible generation mix $\mathbf g$ is picked by an  IPSO to serve demand. We assume that at least one feasible generation mix always exists for every possible load profile. The  IPSO's objective is to pick the cheapest feasible generation mix.   Let us 
 denote the cost of generating $g_v$ units of energy at node $v \in \mathcal{B}$ as a strongly convex and continuous function  $c_v(g_v)$, and the vector of generation costs as $\boldsymbol{c} (\mathbf{g}) = [c_v(g_v)]_{v \in \mathcal{B}}$.
Given  a demand $ \mathbf d$ from EVs, the IPSO  solves an economic dispatch problem to decide the optimal generation dispatch $\mathbf g*$ \cite{glover2011power}:
\begin{align} \label{soc.ipso}
\mathbf g^* = &~\mbox{argmin}_{  \mathbf{g}}   \mathbf{1}^T \boldsymbol{c} (\mathbf{g})   \\& ~\mbox{s.t.}~~  \mathbf g^{\min} \preceq \mathbf g \preceq \mathbf g^{\max}, \nonumber~\mathbf{1}^T  ( \mathbf d + \mathbf u - \mathbf g)  =  0,\nonumber\\ &~~~~~   \mathbf H ( \mathbf d + \mathbf u - \mathbf g)  \preceq \mathbf c.  \nonumber
\end{align}

Note that the optimal traffic and generation schedules determined through \eqref{soc.itso} and \eqref{soc.ipso} minimize the total cost to society. However, they do not necessarily minimize the cost of each individual entity that is involved, e.g., the EV drivers or the generators. Hence, one cannot merely ask these selfish users to stick to the socially optimal schedule. An economic mechanism is necessary to align selfish behavior with socially optimal resource consumption behavior. The use of pricing mechanisms is a way of achieving this goal in a  distributed and incentive-compatible fashion. We highlight pricing mechanisms that can be used for  \eqref{soc.itso} and \eqref{soc.ipso}  next.

\vspace{-.4cm}

\subsection{Pricing Mechanism for Electric Power}

To incentivize profit-maximizing generators to produce at an output level $g_v$, we apply the principle of marginal cost pricing. The principle states that the electricity price at node $v$,  i.e., $p_v$, must satisfy:
\begin{equation}
\frac{\partial c_v(g_v)}{\partial g_v} = p_v \rightarrow
\nabla_{\mathbf{g}} \mathbf{1}^T   \boldsymbol{c} (\mathbf{g}) = \mathbf p.
\end{equation}
Let us introduce Lagrange multipliers $\gamma$ and $\boldsymbol \mu$ respectively for the balance and line flow constraints in  \eqref{soc.ipso}.
Writing the KKT stationarity condition for \eqref{soc.ipso} then leads to:
\begin{equation}\label{lmpdef}
 \mathbf p = \gamma  \mathbf{1} +  \mathbf H^T \boldsymbol \mu,
\end{equation}
 commonly referred to as Locational Marginal Prices (LMP) in the power system literature. 
The reader should note that this is the same price vector $\mathbf p$ that is fed into the ITSO optimization \eqref{soc.itso} and would affect the charging demand at different nodes of the grid, i.e., $ \mathbf d$ in \eqref{soc.ipso}, which would in turn affect the price $\mathbf p$ again. This feedback loop highlights the coupling between smart power and transportation systems that we are interested in, further studied in Section V.
\vspace{-.3cm}
\subsection{Tolls to Align Selfish User Behavior with Social Optimum}
In the transportation network, if every user solves an ESPP given in \eqref{ind.opt2} and no tolls are imposed by the ITSO ($\theta_a = 0, \forall a \in \mathcal{A}\cup  \mathcal{C}$), the aggregate flow would be determined based on a state of user-equilibrium. This user equilibrium flow is most likely not equivalent to the social optimal flow in \eqref{soc.itso}.  To mathematically characterize this equilibrium,  define an auxilary modified cost function  $\boldsymbol{s}^e(\boldsymbol{\lambda}) = [s^e_a(\lambda_a)]_{a\in \mathcal L}$ for the extended transportation graph's arcs as:
\begin{equation}\label{aux}\boldsymbol{s}^e(\boldsymbol{\lambda}) = \boldsymbol{s}(\boldsymbol{\lambda}) + \mathbf M^T \mathbf p,\end{equation} with  $ \boldsymbol{s}(\boldsymbol{\lambda})$ and $\mathbf M$ given by \eqref{costslambda} and \eqref{mdef}.  Moreover, let $\boldsymbol{w}^e(\boldsymbol{\lambda}) = [\int_{x=0}^{\lambda_a}s^e_a(x) dx]_{a\in \mathcal L}$.
Then, according to the well-known Wardrop's first principle  \cite{yildirim2001congestion}, the user equilibrium flow would be the solution of the optimization problem:
\begin{align} \label{ue.itso}
&\min_{\mathbf{f}_q }  \mathbf{1}^T\boldsymbol{w}^e(\boldsymbol{\lambda})  \\& \mbox{s.t.}~~ \mbox{Constraints $(\star)$ in \eqref{soc.itso}}\nonumber
\end{align}
So how can the ITSO get the individual drivers to follow the socially optimal flow calculated in \eqref{soc.itso}? We present the answer in Theorem \ref{mp}.

\begin{theorem}\label{mp}(Marginal congestion pricing)
The aggregate effect of individual route and charge decisions made by EV drivers, i.e., the solution of \eqref{ue.itso}, will
be equivalent to the optimal social charge and route decision  in \eqref{soc.itso} iff the ITSO  imposes a toll 
$\boldsymbol{\vartheta} = [\vartheta_a]_{ a\in{\mathcal L}}$ at each arc of the extended graph  equal to the externality introduced by each user that travels the arc on the other users' costs, i.e.,
\begin{equation}\boldsymbol \vartheta = \nabla(diag(\boldsymbol{\lambda})\boldsymbol{s}(\boldsymbol{\lambda})) - \boldsymbol{s}(\boldsymbol{\lambda}).\end{equation}
\end{theorem}
\begin{proof}
Using the definition of the modified cost vector $\boldsymbol{s}^e(\boldsymbol{\lambda})$, \eqref{soc.itso} can be written as a classic traffic assignment problem:
\begin{align} \label{soc.itso2}
&\min_{\mathbf{f}_q } \boldsymbol{\lambda}^T \boldsymbol{s}^e(\boldsymbol{\lambda})  \\& \mbox{s.t.}~~ \mbox{Constraints $(\star)$ in \eqref{soc.itso}}\nonumber
\end{align}
The result then  follows from applying classic results on Wardrop's first and second  principles \cite{yildirim2001congestion} to the extended graph, acknowledging that  $\nabla(diag(\boldsymbol{\lambda})\mathbf b) = \mathbf b$.\end{proof}
%
%\footnote{Wardrop's first and second principles respectively characterize a state of ``user equilibrium'' and ``social optimum'' for the traffic assignment problem.  The extra tolls imposed will  induce the equilibrium flow of the game among individual EV users to become socially optimal.} 

\begin{remark} (Congestion Mark-up at Charging Stations) The arc toll $\theta_a$ on the virtual charging station entrance arcs  would correspond to a congestion mark-up  (plug-in fee) for all EVs stopping at each station. 
 This captures the user externality introduced by limited charging station capacity. The spots at FCSs located at busy streets and highways   or ones that allow a user to take  less congested routes are coveted by many drivers and thus  have higher plug-in fees.
\end{remark}

\section{Interactive Network Operation}\label{collab}
 For scalability reasons, the IPSO cannot be expected to consider detailed models of the transportation system demand flexibility when calculating  the prices $\mathbf p$.  However, we next  show that completely ignoring the interconnection between the two infrastructures (the status-quo) can have adverse effects on both infrastructures.  This motivates us to  introduce a collaborative pricing scheme using dual decomposition. The schemes we study for interactions between the IPSO and ITSO towards network operation  are highlighted in Fig. \ref{schemes}.

\subsection{Greedy pricing}\label{greedy}
Let us look at the scenario that would happen if no corrective action is taken in regards to how the grid is operated today and hence,  smart transportation and energy systems are  operated separately. In this disjoint model, the IPSO ignores the fact that the  load due to EV charge requests can move from one grid bus to another in response to posted prices. Instead, LMPs  are designed assuming that charge events will happen exactly as in the last period (this could be the previous day, the average of the previous month, etc.). On the other hand, the  ITSO  ignores the effect of EV charge requests on electricity prices, and takes electricity prices as a given when designing road and FCS congestion tolls.

{\bf Claim}: Under this greedy pricing scheme, the congestion and electricity prices $\boldsymbol \theta$ and $\mathbf p$ could oscillate indefinitely.  
 
We substantiate this claim through a numerical example in Section \ref{sec.numerical}. This, along with the loss of welfare experienced when our infrastructure is operated at a suboptimal state, motivates us to look into schemes which can allow the ITSO and IPSO to operate their networks optimally and reliably.

\subsection{Collaborative pricing}
\begin{proposition}\label{optdd}
An efficient market clearing LMP $\mathbf p$ can be posted through a ex-ante collaboration between the IPSO and ITSO following a dual decomposition algorithm.  
\end{proposition}
\begin{proof}
A market clearing price is efficient (maximizes social welfare) if the  flow and generation  values  $\boldsymbol{\lambda}^*$ and $\mathbf g^*$  are the solution of:
\begin{align} \label{soc.joint}
&\min_{\mathbf{f}_q, \mathbf{g}} \boldsymbol{\lambda}^T \boldsymbol{s}(\boldsymbol{\lambda}) + \mathbf{1}^T \boldsymbol{c} (\mathbf{g})~~~~~~~~~~~~~~~~~~~~~~~~~~~~~~~~ \end{align} 
\vspace{-1.3cm}
\begin{multicols}{2}
\begin{equation}
\mbox{s.t.}
 \begin{cases} \mathbf{f}_q  \succeq \mathbf 0,\\ \mathbf{1}^T \mathbf{f}_q = m_{q}, \\ \boldsymbol \lambda = \sum_{q\in \mathcal{Q}}  \mathbf{A}_q  \mathbf{f}_q,  \end{cases} \nonumber
\end{equation}
\break
\begin{equation}
\begin{cases} \mathbf g^{\min} \preceq \mathbf g \preceq \mathbf g^{\max}, \\ \mathbf{1}^T  (\mathbf M  \boldsymbol{\lambda} + \mathbf u - \mathbf g)  =  0,\\ \mathbf H (\mathbf M  \boldsymbol{\lambda} + \mathbf u - \mathbf g)  \preceq \mathbf c.  \end{cases} \nonumber
\end{equation}
\end{multicols}
%
% \begin{cases} \mathbf{f}_q  \succeq \mathbf 0,\\ \mathbf{1}^T \mathbf{f}_q = m_{q}, \\ \boldsymbol \lambda = \sum_{q\in \mathcal{Q}}  \mathbf{A}_q  \mathbf{f}_q,  \end{cases} \nonumber\\&~~~~~  \begin{cases} \mathbf g^{\min} \preceq \mathbf g \preceq \mathbf g^{\max}, \\ \mathbf{1}^T  (\mathbf M  \boldsymbol{\lambda} + \mathbf u - \mathbf g)  =  0,\\    \mathbf H (\mathbf M  \boldsymbol{\lambda} + \mathbf u - \mathbf g)  \preceq \mathbf c.  \end{cases} \nonumber

\begin{figure}[t]
\centering
\includegraphics[width = 0.9 \linewidth]{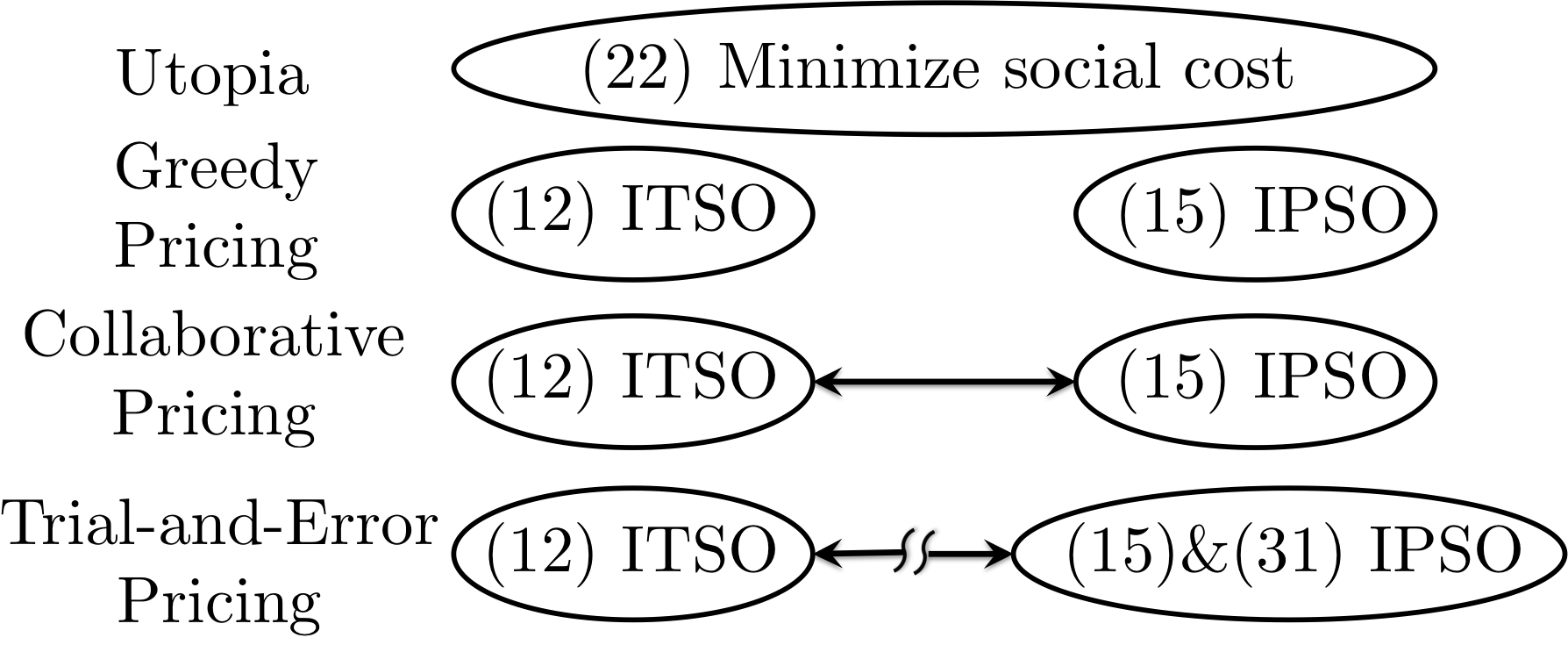}
\caption{The different network operation schemes studied} 
\label{schemes}
\end{figure}

The last two constraints contain both decision variables and  couple the IPSO and ITSO optimization problems. Let us introduce Lagrange multipliers $\gamma$ and $\boldsymbol \mu$ respectively for the balance and line flow constraints. The partial Lagrangian of   \eqref{soc.joint} considering only the coupling constraints is:
\begin{align}  L(\mathbf{f}_q|_{q\in \mathcal Q}, \mathbf{g}, \gamma, \boldsymbol \mu) &\!=\!  \boldsymbol{\lambda}^T \boldsymbol{s}(\boldsymbol{\lambda}) \!+\! \mathbf{1}^T \boldsymbol{c} (\mathbf{g})   \!+\! \gamma \mathbf{1}^T  (\mathbf M  \boldsymbol{\lambda} + \mathbf u - \mathbf g) \nonumber\\& +\boldsymbol \mu^T ( \mathbf H (\mathbf M  \boldsymbol{\lambda} + \mathbf u - \mathbf g)  - c) \end{align}
with $\boldsymbol \mu  \succeq \mathbf 0$. Since $  L(\mathbf{f}_q|_{q\in \mathcal Q}, \mathbf{g}, \gamma, \boldsymbol \mu)$ is separable, we can minimize over $\mathbf{f}_q|_{q\in \mathcal Q}$ and $\mathbf{g}$ in two separate subproblems,  allowing us to use  standard dual decomposition with projected subgradient methods to find the optimal price. Consider a sequence $\{\gamma^{(k)} \}$ and $\{\boldsymbol \mu^{(k)}\}$ of Lagrange multipliers generated by the iterative decomposition scheme. Then, at the $k$-th iteration,   the ITSO solves for the optimal extended graph flow $\boldsymbol{ \lambda}^{(k)}$ through a subproblem that has the same structure as \eqref{soc.itso}, with electricity prices at iteration $k$, $\mathbf p^{(k)}$, set as:
\begin{align}  \label{priceiteration}
&\mathbf p^{(k)}  =  \gamma^{(k)}  \mathbf{1} +  \mathbf H^T \boldsymbol \mu^{(k)}.
\end{align}
On the other hand, the subproblem solved by the IPSO is:
\begin{align}  
 \mathbf g^{(k)}  =~ & \mbox{argmin}_{  \mathbf{g}}   \mathbf{1}^T \boldsymbol{c} (\mathbf{g})  -  (\gamma^{(k)}  \mathbf{1} +  \mathbf H^T \boldsymbol \mu^{(k)})^T \mathbf{g}  \\& \mbox{s.t.}~~  \mathbf g^{\min} \preceq \mathbf g \preceq \mathbf g^{\max}. \nonumber
\end{align}
The IPSO then updates the balance and congestion components of the LMP, i.e., $\gamma^{(k)} $ and $\boldsymbol \mu^{(k)}$,  through:
\begin{equation}\label{updatedual}
\left( \begin{array}{c}
\gamma^{(k+1)}\\
\boldsymbol \mu^{(k+1)}
\end{array} \right) \!=\! \left( \!\!\!\begin{array}{c}
\gamma^{(k)} + \alpha_k ( \mathbf{1}^T  (\mathbf M  \boldsymbol{\lambda}^{(k)} + \mathbf u - \mathbf g^{(k)})) \\
\{\boldsymbol \mu^{(k)} + \alpha_k ( \mathbf H (\mathbf M  \boldsymbol{\lambda}^{(k)} + \mathbf u - \mathbf g^{(k)})  - c)\}^+
\end{array} \!\!\!\!\right) 
\end{equation}

It is shown that with a small enough step size, the dual decomposition method converges to the solution of \eqref{soc.joint} \cite{shor1985minimization}. Hence, if the electricity price  $\bf p$ is   $\gamma^{\star}  \mathbf{1} +  \mathbf H^T \boldsymbol \mu^{\star} = \lim_{k \rightarrow \infty} \gamma^{(k)}  \mathbf{1} +  \mathbf H^T \boldsymbol \mu^{(k)}$, the market clears and the generator outputs and system flow will be equal to $\mathbf g^*$ and $\boldsymbol{\lambda}^*$.
\end{proof}

\vspace{-.4cm}
\subsection{Optimal reserve capacity for trial-and-error pricing}
In theory, the above algorithm can eliminate the need for the existence of an ex-ante\footnote{The term ex-ante refers to actions that are adjusted as a result of forecasting user behavior and not actual observations, while ex-post refers to actions that are based on actual observations rather than forecasts. } ITSO collaboration for calculating electricity prices. Instead, imagine that the IPSO can actually post electricity prices according to  \eqref{updatedual} and  gradually find the optimal market clearing LMP by observing the charging demand of the actual transportation system\footnote{Note that this is only possible if the time-scale at which the network flow $\boldsymbol \lambda$ reaches its new equilibrium in response to  new posted prices ${\bf p}$ and tolls $\boldsymbol \theta$ is much smaller that the time-scale at which electricity costs or travel demands change.}. When dealing with unknown demand functions in commodity pricing, this is referred to as {\it the trial-and-error approach}, see, e.g., \cite{yang2010road}, for prior use of such approaches in toll design.

Implementing this approach has two requirements:

1) The IPSO should be willing to move away from the greedy pricing scheme  in order to eventually maximize  societal welfare (even though the extra welfare generated might not be easily quantifiable and the  operating point might not correspond to minimum generation costs);

2) More importantly, primal feasibility is most likely violated when using Lagrangian relaxation to handle coupling constraints in \eqref{soc.joint}. This means that when posting prices according to \eqref{updatedual}, the IPSO should expect the balance and flow constraints to be violated in order for the algorithm to converge, and plan accordingly. In power grids, any unpredictable violation of reliable system operation is referred to as a contingency (a threat to the security of the system) and is handled through {\it generation reserves}.

\begin{definition}{(Reserves)}  In power grids, a generation reserve capacity of ${\bf r} = [r_v]_{v\in \mathcal B}$ allows the IPSO to compensate for any demand-supply imbalances after market clearing {\blue as long as $y_v \in [-r_v,r_v]$ , or equivalently $-\mathbf r \preceq \mathbf y \preceq \mathbf r$}. This is typically done by adjusting the output of an already online generator either upward or downward. Given a reserve capacity of $\mathbf r$, the balance equation will become:
\begin{equation} \label{balancecons}
\mathbf{1}^T  (\mathbf d + \mathbf u - \mathbf g - \mathbf y) = 0,~~-\mathbf r \preceq \mathbf  y \preceq \mathbf r \end{equation}
where $\mathbf y$ can be chosen at the IPSO's discretion after observing the demand $\mathbf d$.
 The reserve capacity $\mathbf r$ should be procured in advance.
\end{definition}
Note that the dispatched reserve generation $ \mathbf y$ affects the line flows and hence the flow constraint becomes:
\begin{equation} \label{flowcons}
 \mathbf H (\boldsymbol{d} + \mathbf u - \mathbf g - \mathbf y)  \preceq \mathbf c.
\end{equation}

Here we will use 
bounds on primal infeasibility to determine the reserve capacity $\mathbf r$ that needs to be procured by the IPSO in this type of ex-post LMP adjustment. {\color{black} For simplicity, we consider a constant step size rule such that $\alpha_k = \alpha$ for all $k$.} Note that in this scheme, after each price adjustment iteration $\boldsymbol{p}^{(k)}$, the approximate primal solutions, which are the last iterate $\boldsymbol{ \lambda}^{(k)}$ and $\boldsymbol{ g}^{(k)}$, are actually implemented. Assume that the IPSO knows that during the $k$-th iteration, 
\begin{equation}\label{bound}
\left( \begin{array}{c}
|\mathbf{1}^T  ( \boldsymbol{d}^{(k)} + \mathbf u - \mathbf g^{(k)})| \\
   \mathbf H (\boldsymbol{d}^{(k)} + \mathbf u - \mathbf g^{(k)})  - \mathbf c
\end{array} \right)  \preceq  \left( \begin{array}{c}
a_k \\  \mathbf w_k
\end{array} \right),
\end{equation}
where $d^{(k)} = \mathbf M  \boldsymbol{\lambda}^{(k)}$.
For dual first order algorithms such as dual gradient and dual fast gradient methods (any of which can be employed by the IPSO for price update), such bounds were recently provided by \cite{necora}. For example, for dual gradient methods, one possible bound is given by:
% \begin{equation} \label{eq:primal_infeas}
%   \left( \begin{array}{c}
%a_k \\  \mathbf w_k
%\end{array} \right) =  \left\|   \left( \begin{array}{c}
%\mathbf{1}^T \\  \mathbf H 
%\end{array} \right) \right\| \|\mathbf g^{(0)}  - \mathbf g^* \|\mathbf{1}.
%\end{equation}
%To calculate this bound,  the IPSO has to have access to the set of {\it potentially optimal}  generation dispatch  $\mathbf g^*$. 
%Without access to travel patterns,  we can rely on the fact that $\mathbf p$ is a convex function of $\mathbf d$, and that   $\nabla \mathbf{1}^T \boldsymbol{c} (\mathbf{g}) = \mathbf p$, to conclude that $\mathbf g^*$ is constrained to lie in  a box set if $\mathbf d$ lies in a box set. The extreme points of this set  can be evaluated  by setting ${\bf d}$ equal to its upper and lower bounds given the limited capacity of charging stations (all charging stations empty versus   all charging stations operating at full capacity). %provided that the derivative of the cost function $c_{\upsilon}(g)$ is non-decreasing in $g$. 
 \begin{equation} \label{eq:primal_infeas}
   \left( \begin{array}{c}
a_k \\  \mathbf w_k
\end{array} \right) =   \frac{3 \left\|  \left( \begin{array}{c}
\gamma^{(0)} \\  \boldsymbol{\mu}^{(0)}
\end{array} \right) -  \left( \begin{array}{c}
\gamma^{\star} \\  \boldsymbol{\mu}^{\star} 
\end{array} \right) \right\|_2 }{\alpha \sqrt{k}} \mathbf{1},
\end{equation}
{\newblue where $\alpha = \alpha_k \leq 1/L_d, \forall k$ and $L_d$ is the Lipschitz constant for the dual problem of \eqref{soc.joint}. We now need to show that the bound in \eqref{eq:primal_infeas} is well-defined.
\begin{lemma}
The dual problem of \eqref{soc.joint} has finite Lipschitz constant  and bounded optimal dual variables, i.e., $L_d < \infty$ and
$\| (\gamma^\star, \bm{\mu}^\star) \| < \infty$. 
\end{lemma}
\begin{proof}
The finiteness of $L_d$ is a consequence of the strong convexity of the objective function\footnote{Recall that $s_a(\lambda_a)$ is non-negative, convex and increasing (cf.~\eqref{incon}), therefore the product $\lambda_a s_a(\lambda_a)$ is strongly convex.} of \eqref{soc.joint} \cite{necora}. 
Furthermore, we see that \eqref{soc.joint} is a convex problem with linear inequality constraints and its optimal 
objective value is finite (as we have assumed that at least one feasible generation mix exists for every possible load profile). Consequently, strong duality holds for \eqref{soc.joint} and there exists a set of  
bounded optimal dual variables. 
\end{proof}}

%As the objective function of \eqref{soc.joint} is strongly-convex, the constant $L_d$ is finite and a non-trivial lower bound to the constant step size $\alpha$ can be found. 
To calculate \eqref{eq:primal_infeas}, the IPSO has to access to the set of {\it potentially optimal} energy and congestion prices $\gamma$ and $\boldsymbol \mu$. 
%Notice that as the objective function of \eqref{soc.joint} is strongly convex and as   Slater's condition holds, a finite optimal dual variable $(\gamma^\star, \bm{\mu}^\star)$  exists \cite{necora}. Therefore, the right hand side of  \eqref{eq:primal_infeas} is  bounded.
Without access to travel patterns,  an estimate of the upper bound to $\| ( \gamma^\star, \bm{\mu}^\star ) \|_2$ can be evaluated using many methods. For example, we can calculate these bounds by a) performing Monte-Carlo simulations by setting ${\bf d}$ to different values between its upper and lower bounds given the limited capacity of charging stations (charging stations  operating anywhere until full capacity) \cite{5130234}; b) studying critical load levels as suggested in \cite{li2009congestion}; c) solving a Mathematical Program with Equilibrium Constraints (MPEC) \cite{7192738}. {\newblue We chose to solve an MPEC to find this bound; see Appendix~\ref{app:mpec}.}

%Without access to  travel patterns, this set could be calculated by varying the load of each charging station between zero and its capacity and hence calculating the set of potentially optimal prices in a Monte-Carlo fashion {\color{black} (Really, can we do that? I think it's equivalent to saying that you can use brute-force to solve any optimization problem in the world..., which is kind of true though)}.

 Given \eqref{bound}, the IPSO needs to ensure that \eqref{balancecons} and \eqref{flowcons} hold by appropriately choosing a reserve capacity $\mathbf r_k$ to be procured for each future $k$ to ensure system security.

\begin{proposition} (Security-Constrained Ex-post Price Adjustments) Given unit reserve capacity prices $\boldsymbol \xi_k$ at each node of the power grid for iteration $k$, the optimal reserve capacity $\mathbf r^\star_k$ to be procured at different grid buses for iteration $k$ of the price adjustment algorithm is given by:
 \begin{align}\label{robustcap}
&  {\bf r}_k^\star = \arg \min_{ {\bf r}_k  }  \bm{\xi}_k^T {\bf r}_k \\
{\rm s.t.} & \max_{\bm{\eta}^j, \forall j,  \boldsymbol{\theta}^i, \forall i} \!\! -  \theta_1^i \mathbf 1^T \bm{\eta}^j +    (\mathbf H \bm{\eta}^j  - \mathbf c)^T \boldsymbol{\theta}^i_2  -  \mathbf r_k^T ( \boldsymbol{\theta}^i_3 +  \boldsymbol{\theta}^i_4) \leq 0,
\nonumber\end{align}
where the constraint is a piecewise-linear function of ${\bf r}_k $, and the numbers $(  \boldsymbol{\theta}^i, \bm{\eta}^j), i = 1,\ldots, I, j = 1,\ldots, J$ are given.
\end{proposition}
\begin{proof} The optimal reserve capacity $\mathbf r_k$ is equal to the cheapest possible nodal reserve capacity combination that can restore the network balance and flow constraints under {\it any} possible amount of feasibility violation specified in \eqref{bound}, i.e., we have the following robust optimization problem:
 \begin{align}
& \mathbf r^*_k =   \mbox{argmin}_{\mathbf r_k} \boldsymbol \xi_k^T \mathbf r_k \label{robust1} \\
 &{\rm s.t.}  \forall~\bm{\eta} \in {\cal N}, \exists \mathbf y: |\mathbf y| \preceq \mathbf r_k: \!\begin{cases} {\bf 1}^T (\bm{\eta} - {\bf y}) = 0 \\
{\bf H} (\bm{\eta} - {\bf y}) \preceq {\bf c}  \end{cases} \nonumber
\end{align}
where $\bm{\eta} = {\bf d} + {\bf u} - {\bf g}$ and 
\[ \begin{array}{l} {\cal N} = 
\{ \bm{\eta} : |{\bf 1}^T \bm{\eta}| \leq a_k,~{\bf H}\bm{\eta} - {\bf c} \preceq {\bf w}_k, \bm{\eta}_{min} \preceq \bm{\eta} \preceq \bm{\eta}_{max} \},\end{array} \]
where $\bm{\eta}_{min}$ and $\bm{\eta}_{max}$ denotes the minimum/maximum possible $\bm{\eta}$. 
 Problem \eqref{robust1} is equivalent to:
 \begin{align} \label{reservecap}
& \mathbf r^*_k =   \mbox{argmin}_{\mathbf r_k} \boldsymbol \xi_k^T \mathbf r_k \\
&\mbox{s.t.}~  \forall~\bm{\eta} \in {\cal N}: \mbox{$Q(\bm{\eta},\mathbf r_k)$ is feasible} \nonumber 
\end{align}
where \begin{align}Q(\bm{\eta},\mathbf r_k) =   &\min_{\mathbf y}~ 0\\ \mbox{s.t.}~ &\mathbf 1^T \mathbf y = \mathbf 1^T \bm{\eta}\nonumber,~ -\mathbf H \mathbf y \preceq \mathbf c - \mathbf H \bm{\eta} \nonumber \\&- \mathbf r_k \preceq \mathbf y  \preceq \mathbf r_k  \nonumber\end{align}
If $Q(\bm{\eta},\mathbf r_k)$ is feasible, its dual problem $Q^*(\bm{\eta},\bf r_k)$ is bounded and the dual optimum will be 0. Thus, we can  write \eqref{reservecap} as:
 \begin{align} \label{reservecap2}
& \mathbf r^*_k =   \mbox{argmin}_{\mathbf r_k} \boldsymbol \xi_k^T \mathbf r_k \\
&\mbox{s.t.}~ \forall~\bm{\eta} \in {\cal N}:  \max_{\boldsymbol{\theta} \in \mathcal T}  F(\bm{\eta},\boldsymbol{\theta},\bf r_k) = 0, \nonumber
\end{align}
where
\begin{align}\mathcal T =  \{&  \boldsymbol{\theta} = (\theta_1, \boldsymbol{\theta}_2,  \boldsymbol{\theta}_3,  \boldsymbol{\theta}_4) | \boldsymbol{\theta}_2,  \boldsymbol{\theta}_3,  \boldsymbol{\theta}_4 \succeq  \boldsymbol{0},  \\ &\theta_1 \mathbf 1 - \mathbf H^T  \boldsymbol{\theta}_2 -  \boldsymbol{\theta}_3 +  \boldsymbol{\theta}_4 = \mathbf 0 \}, \nonumber \\ F(\bm{\eta},\boldsymbol{\theta},\bf r_k) = &-  \theta_1 \mathbf 1^T \bm{\eta} +    (\mathbf H \bm{\eta}  - \mathbf c)^T \boldsymbol{\theta}_2 -  \mathbf r_k^T ( \boldsymbol{\theta}_3 +  \boldsymbol{\theta}_4). \nonumber\end{align}
Since $\mathbf 0 \in \mathcal T$, this is equivalent to:
 \begin{align} \label{reservecap3}
&\mathbf r^*_k =   \mbox{argmin}_{\mathbf r_k} \boldsymbol \xi_k^T \mathbf r_k \\
&\mbox{s.t.}~  \max_{ \bm{\eta} \in {\cal N},  \boldsymbol{\theta} \in \mathcal T} F(\bm{\eta},\boldsymbol{\theta},\bf r_k) \leq 0.\nonumber
\end{align}
Note that $F(\bm{\eta},\boldsymbol{\theta},\bf r_k)$ is neither convex  nor concave in $\bm{\eta}$ and $  \boldsymbol{\theta}$ (bilinear). The constraint set $\bm{\eta} \in {\cal N},  \boldsymbol{\theta} \in \mathcal T$ is a polyhedron and hence, the  optimal solution of $\max_{ \bm{\eta} \in {\cal N},  \boldsymbol{\theta} \in \mathcal T} F(\bm{\eta},\boldsymbol{\theta},\bf r_k)$ is one of the extreme points $(  \boldsymbol{\theta}^i, \bm{\eta}^j), i = 1,\ldots, I, j = 1,\ldots, J$ of the polyhedrons ${\cal T}$ and ${\cal N}$. This shows that the constraint is a  convex piecewise linear function
in ${\bf r}_k$.
\end{proof}

\begin{remark}
In general, we have no knowledge of the
extreme points of ${\cal N}$ and ${\cal T}$, and computing $ \max_{ \bm{\eta} \in {\cal N},  \boldsymbol{\theta} \in \mathcal T} F(\bm{\eta},\boldsymbol{\theta},\bf r_k)$ is non-trivial. Hence, proper approximation algorithms need to be studied for solving \eqref{robustcap}. However, this is out of the scope of this paper. See \cite{bertsimas2013adaptive} for the treatment of a somewhat similar problem, where the use of an outer approximation algorithm is proposed.  Instead, in our numerical results, we resort to a sample/scenario-approximation method \cite[Chapter 2.6]{bental}. For example, we replace the set ${\cal T} \times {\cal N}$ by a finite set $\{ ( \bm{\theta}_i,\bm{\eta}_i),~i=1,...,N_s \} \subseteq {\cal T} \times {\cal N}$. In this case, \eqref{reservecap3} will be turned into a convex program with a finite number of  linear inequalities.
\end{remark}

\section{Numerical Examples}\label{sec.numerical}
%{\blue Hoi-To will present an improved version of our previous numerical results in this section.}

This section investigates the need for the joint  EV management scheme we propose through numerical analysis of system performance. We focus on the system level optimization. We assume that charging stations are publicly owned infrastructure for the sake of simplicity. This means that we assume electricity is sold at wholesale prices to EV drivers\footnote{In reality, EV drivers may purchase flat rate charging services from for-profit entities that trade with the wholesale market and can use appropriate economic  incentives (similar to the tolls discussed in this paper) and recommendation systems to guide the customers towards optimal  stations. This is beyond the scope of this work.}.
%Our aim is to demonstrate the benefits to the society by jointly optimizing in  transportation and power network. 

The transportation network ${\cal G}$ is shown in Fig.~\ref{fig:transport_graph}. For each arc (road section), we define the latency function as:
\begin{equation}
\tau_a(\lambda_a) = T_a + \lambda_a / 10^4,
\end{equation}
where $T_a$ is the minimum time required to travel through arc $a$. We set $\gamma = \$10^{-3} / {\rm minute}$ for the cost spent en route. Note that this might seem like a rather  low number but it would be scaled up by a factor of 10 if electricity is traded at retail prices instead of wholesale. The power network ${\cal R}$ is modeled using the line and generation cost parameters of the IEEE 9-bus test case, except  that several more buses are modeled as load buses where EVs can charge; also see Fig.~\ref{fig:transport_graph}. 
%For this network ${\cal G}$ the minimum travel cost pattern would be to skip Winters at all as the latter will incur additional travelling cost. However, we also observe that Winters has a much lower base load than in Davis, which implies that the electricity price can be lower in Winters. 
%Consequently, a socially optimal travel pattern may consider diverting some EVs to charge at Winters while sacrificing some additional travelling cost. 

\begin{figure}[t]
\centering
\includegraphics[width = \linewidth]{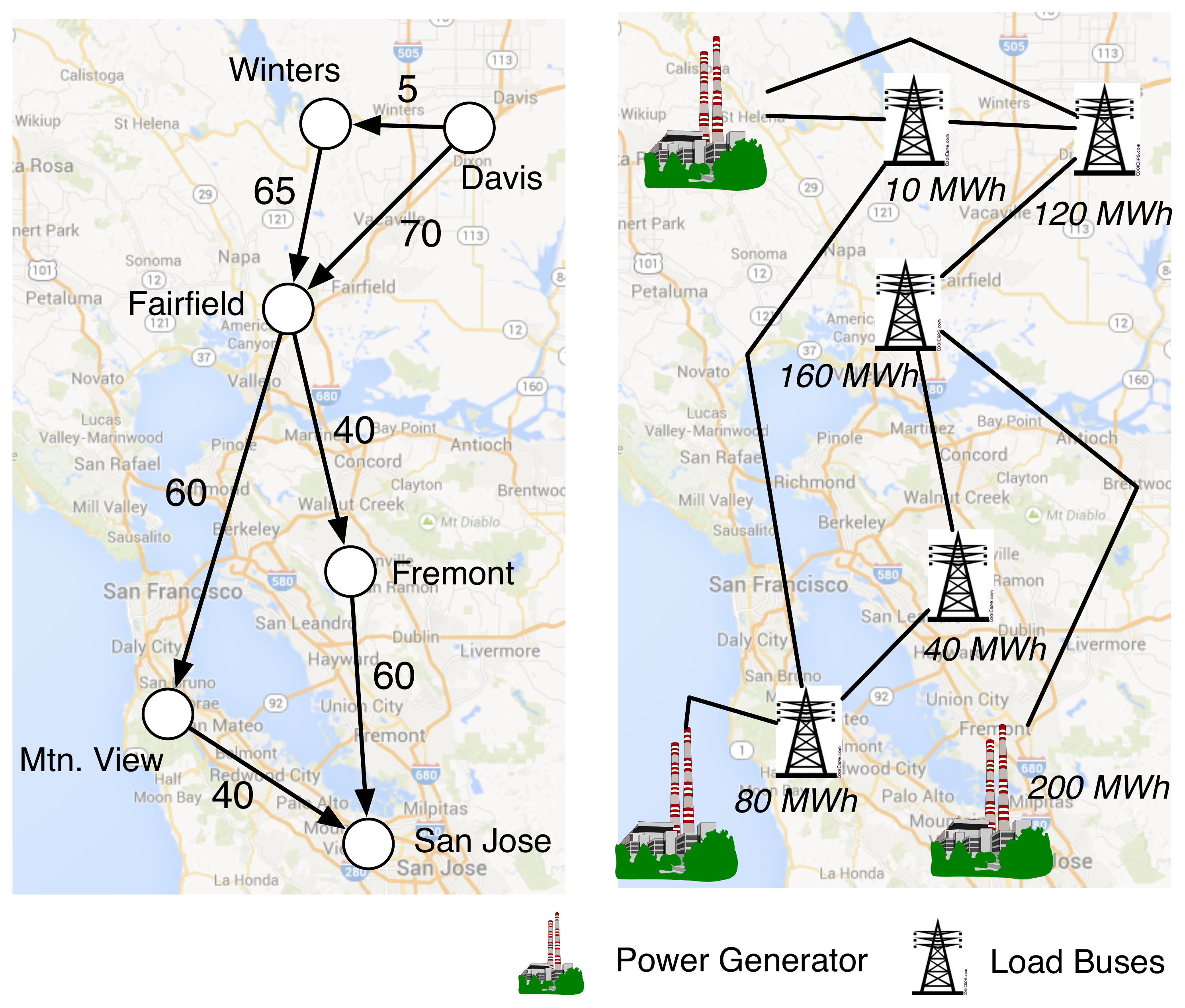}
\caption{(Left) The transportation network ${\cal G}$ for the trip from Davis to San Jose. The value next to each arc denotes the minimum travel time $T_a$ (in minutes); (Right) The power network ${\cal R}$. The base load $u_v$ at each node is denoted in italic. Each of the intermediate node is equipped with an FCS.} 
\label{fig:transport_graph}
\end{figure}

The intermediate nodes, i.e., Winters, Fairfield, Mountain View and Fremont, are equipped with an FCS. Each FCS is capable of supplying 1 kWh to an EV every 5 minutes, and the available charging options are $\{0, 1, 2, 3 \}$ kWh (the same charging options hold for the origin). It is assumed that each EV consumes 1 kWh  of energy to travel  25 miles, and the battery capacity is 6 kWh for all EVs. 
There are two O-D pairs considered in the network. Specifically, $50\%$ of the drivers are traveling from Davis, CA to Mountain View, CA;  and $50\%$ of the drivers are traveling from Davis to San Jose, CA. At the origin, i.e., Davis, the EVs have an initial charge of 4 kWh. As such, there are $|{\cal Q}| = 2$ classes of users.

In the first numerical example, we study how the total number of  EVs can affect the IPSO/ITSO's decision. 
We assume full IPSO/ITSO cooperation such that the social optimal problem \eqref{soc.joint} can be exactly solved. 
As seen in Fig.~\ref{fig:per_traffic}, the traffic pattern changes as we gradually increase the total number of EVs per epoch. For instance, more EVs are routed through Winters, instead of going to Fairfield from Davis directly; similar observations are also made for the Fairfield-Mountain View-San Jose path. This is due to the fact that the power/transportation network has become more congested,  leading to a different traffic pattern.

%An increased number of EVs causes congestion to both the power and transportation network. As such, we see that the optimal IPSO/ITSO decision is also affected.

\begin{figure}[t]
\centering
 \includegraphics[width = 1.0\linewidth]{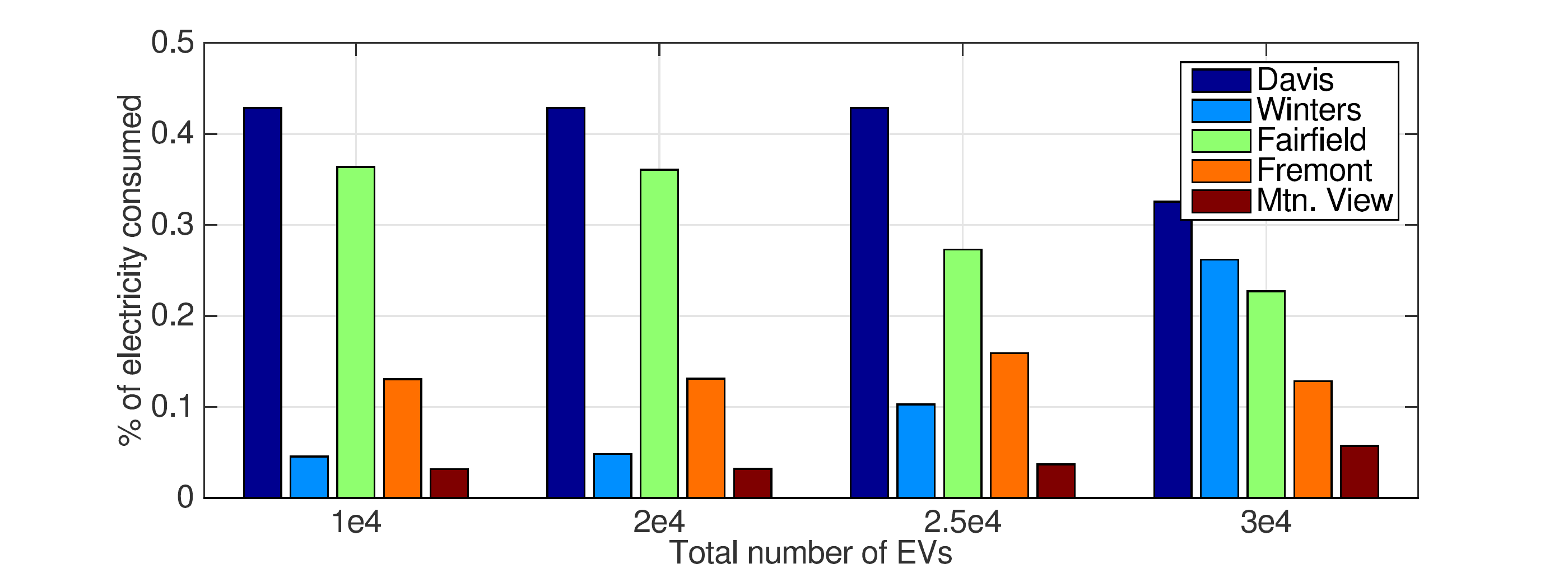}\vspace{.05cm}

 \includegraphics[width = 1.0\linewidth]{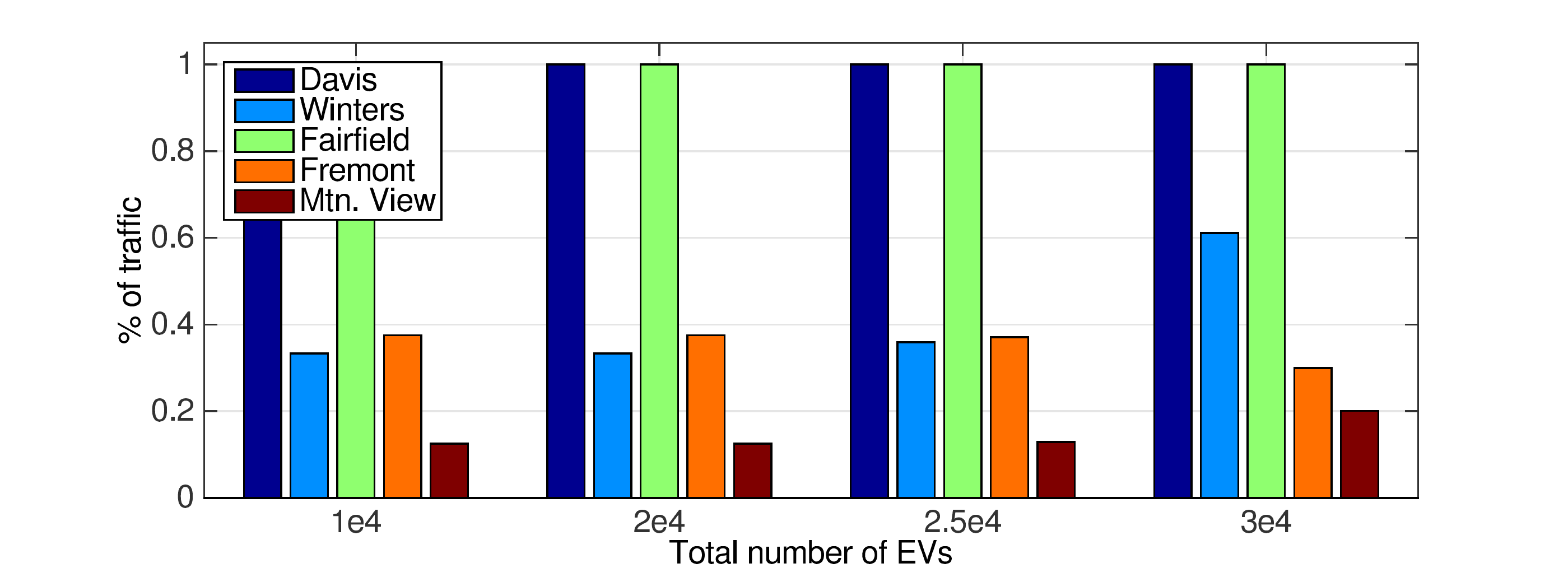}
\caption{Comparing the traffic pattern against the total number of EVs per epoch. (Top) Percentage of electricity consumed at each site; (Bottom) Percentage of traffic leaving from each site.} 
\label{fig:per_traffic}
\end{figure}

% we basically want to show the effects when there are more EVs, then the travel/charging pattern will change 

Our next step is to study the scenario with ex-ante IPSO/ITSO cooperation. We compare the myopic pricing   scheme to the dual decomposition approach.
% In this simulation, the total number of EVs is fixed at $2.26 \times 10^4$ per epoch.  
The first task is to investigate the behavior of the system under the myopic pricing scheme. The total number of EVs is fixed at $2.26 \times 10^4$ per epoch.  In this case, we initialize the electricity price at each site at $\$50$ per MWh to solve \eqref{soc.itso}. The traffic pattern against iteration number of this disjoint optimization is shown in Table~\ref{dpso}. We observe that the system oscillates between two traffic patterns, one having the lower average traveling time and the other one with a lower electricity cost. As described in Section~\ref{greedy}, this oscillation behavior is due to the lack of cooperation between the IPSO and ITSO. We see that at iteration $i=2n$, the electricity prices are the same across the charging stations, therefore the ITSO assigns the traffic by simply minimizing the travel time. This decision, however, leads to an uneven distribution in energy consumption across the power network ${\cal R}$. At iteration $i=2n+1$, the IPSO lowers the electricity price at Winters; and increases the price at Fairfield, Fremont and Mountain View. This motivates the ITSO to re-assign the traffic pattern. 

An interesting point to note is that the disjoint optimization may even lead to an \emph{infeasible} IPSO decision when the total number of EVs considered is large.  This is an extreme case of the example considered in Table I. In this case,  the inability of the greedy pricing method to correctly model the response of the EV population to posted prices would result in an unsafe increase of load at locations where the grid is congested and hence the load needs to be shed to keep transmission lines as well as transformers safe.

%First, we observe that the system oscillates between two traffic patterns, one having the lowest average travel time and the other one diverts all traffic to a charging site with the least power consumption. Secondly, one of the solution is \emph{infeasible} for the IPSO as it demands too much electricity at a certain number of nodes, leading to line flow feasibility violations. This justifies our claim made in Section \ref{greedy}.

\begin{table}[t]
\centering
\caption{Oscillation of traffic patterns with the greedy method.} \label{dpso}
\begin{tabular}{c || c | c | c}
\hline
~ & SO & DO ($i=2n$) & DO ($i=2n+1$) \\
\hline 
\hline
Davis & 67.80 MWh & 67.80 MWh & 67.80 MWh \\ 
~ & @\$57.38/MWh & @\$57.38/MWh & @\$58.04/MWh \\
\hline
Winters & 12.56 MWh & 7.227 MWh & \bfseries 15.87 MWh \\ 
~ & @\$57.38/MWh & @\$57.38/MWh & @\$54.50/MWh \\
\hline
Fairfield & 49.88 MWh & 57.32 MWh & \bfseries 43.52 MWh \\ 
~ & @\$57.38/MWh & @\$57.38/MWh & @\$66.59/MWh \\
\hline
Fremont & 22.56 MWh & 20.83 MWh & \bfseries 25.23 MWh \\ 
~ & @\$57.38/MWh & @\$57.38/MWh & @\$65.09/MWh \\
\hline
Mtn.~View & 5.392 MWh & 5.031 MWh & \bfseries 5,781 MWh \\ 
~ & @\$57.38/MWh & @\$57.38/MWh & @\$61.73/MWh \\
\hline
\hline
Fr.~Winters & 7,533 & 7,534 & \bfseries 7,936 \\ 
\hline
Fr.~Fremont & 8,475 & 8,409 & \bfseries 9,375 \\ 
\hline
Fr.~Mtn View & 2,825 & 2,825 & \bfseries 3,125 \\
\hline
\hline
Travel time & 188.36~min. & 188.36~min. & 188.39~min. \\
\hline
Objective & \$30,332.55 & \$30,364.61 & \$30,333.06 \\
\hline
\end{tabular}
\end{table}

The above example demonstrates that applying myopic pricing may result in an unstable system. Next, we  investigate the performance of the dual decomposition algorithm (cf.~Proposition~\ref{optdd}), which describes a systematic method for cooperation between the IPSO and ITSO. Here, the total number of EVs is fixed at $2.5 \times 10^4$ per epoch. 
The dual decomposition is initialized by setting $\gamma^{(0)} = 57.5$ and $\bm{\mu}^0 = {\bf 0}$. 
% They are determined at the IPSO by estimating the maximum/minimum marginal price in the power grid. 
As the dual decomposition algorithm is known to converge to the social optimum, we are interested in studying its convergence speed with the violation in infeasibility. We set the step size as $\alpha_k = 20$ for all $k$ and apply the algorithm on the same scenario as before. 
{\color{black} For the constant $\| (\gamma^{(0)},\bm{\mu}^{(0)}) - (\gamma^\star,\bm{\mu}^\star)\|_2$ in \eqref{eq:primal_infeas}, we upper bound it by solving an MPEC using an approach similar to \cite{7192738}. }

We compare both infeasibility measures against the iteration number in Fig.~\ref{fig:infeas} (Left). We can see that the dual decomposition algorithm converges in approximately $100$ iterations, and it returns a solution that is approximately feasible. 
%In fact,  convergence speed can be further improved by fine tuning the step size $\alpha_k$. We also compare the \emph{estimated} primal infeasibility given by \eqref{eq:primal_infeas} with the actual infeasibility. 
%Notice that   $\| (\gamma^{(0)},\bm{\mu}^{(0)}) - (\gamma^\star,\bm{\mu}^\star)\|_2$ is approximated as $(\gamma^{\star,max}- \gamma^{\star,min})/2$ due to our choice of $(\gamma^{(0)},\bm{\mu}^{(0)})$. 
We observe a ${\cal O}(1/\sqrt{k})$ decaying trend with the actual infeasibility. 

Lastly, we study the effects of ex-post IPSO price adjustment based on the estimated $({\bf w}_k, a_k)$ in \eqref{eq:primal_infeas}. The reserve procurement problem \eqref{robustcap} is approximated using a sample-approximation method, where the candidate $\bm{\eta}$ points for ${\cal N}$ are selected randomly within the bound $[\bm{\eta}_{min}, \bm{\eta}_{max}]$. We assume that the reserve capacity is purchased at a price of $\$55.00$ per MWh at all sites. Here, an interesting comparison is the \emph{overall} cost needed to purchase such reserve capacity and the cost to operate the system under (the estimated) infeasibility, i.e., the dual objective value. The overall cost is shown in Fig.~\ref{fig:infeas} (Right) as `Dual obj.+Reserve Cost'. We observe that such cost is always higher than the social optimum cost due to a possible mismatch between the electricity cost per unit in purchasing the reserve capacity; yet the difference decreases as the iteration number grows.

\begin{figure}[t]
\centering
\includegraphics[width = 0.48\linewidth]{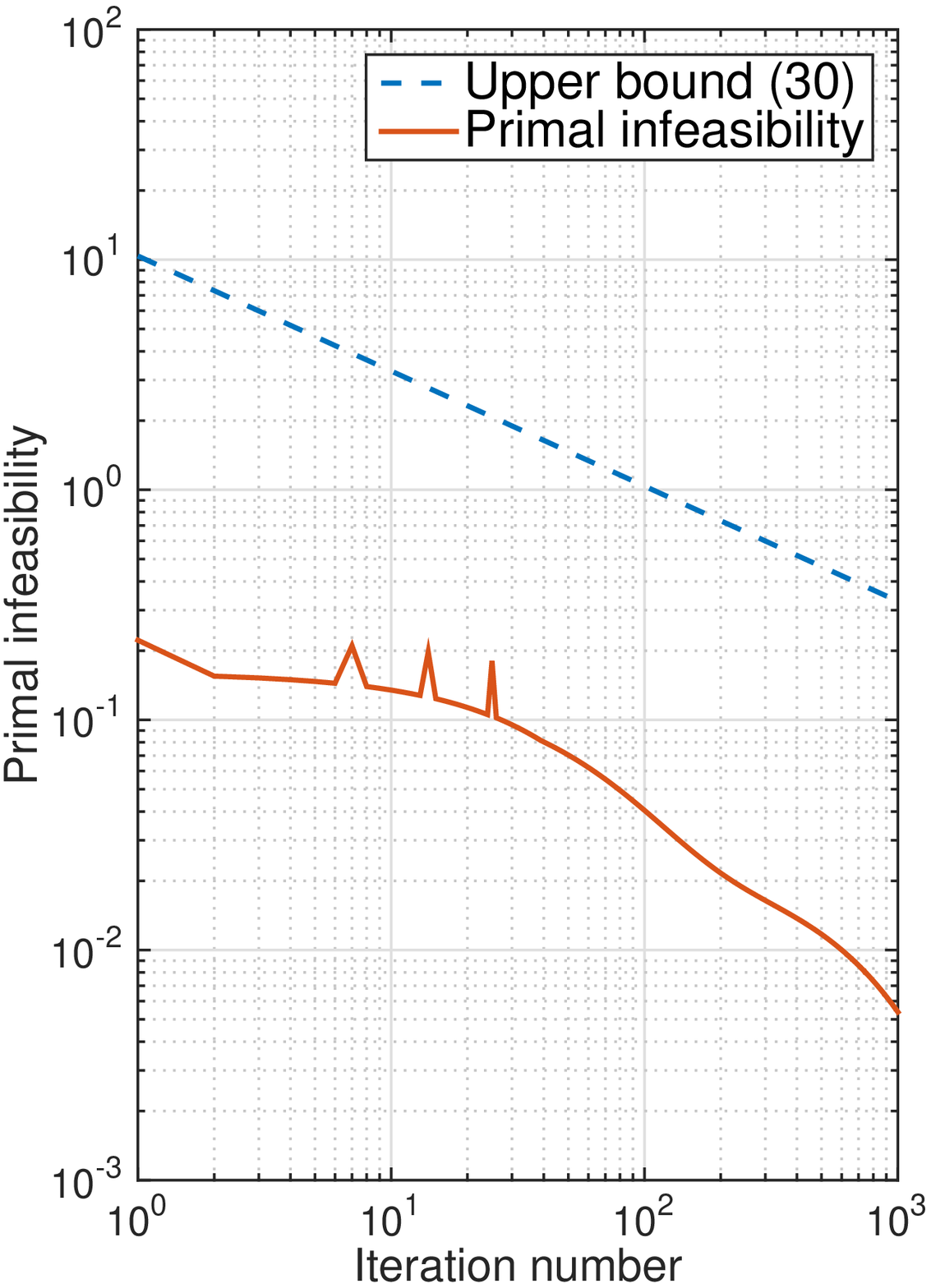} \includegraphics[width = 0.48\linewidth]{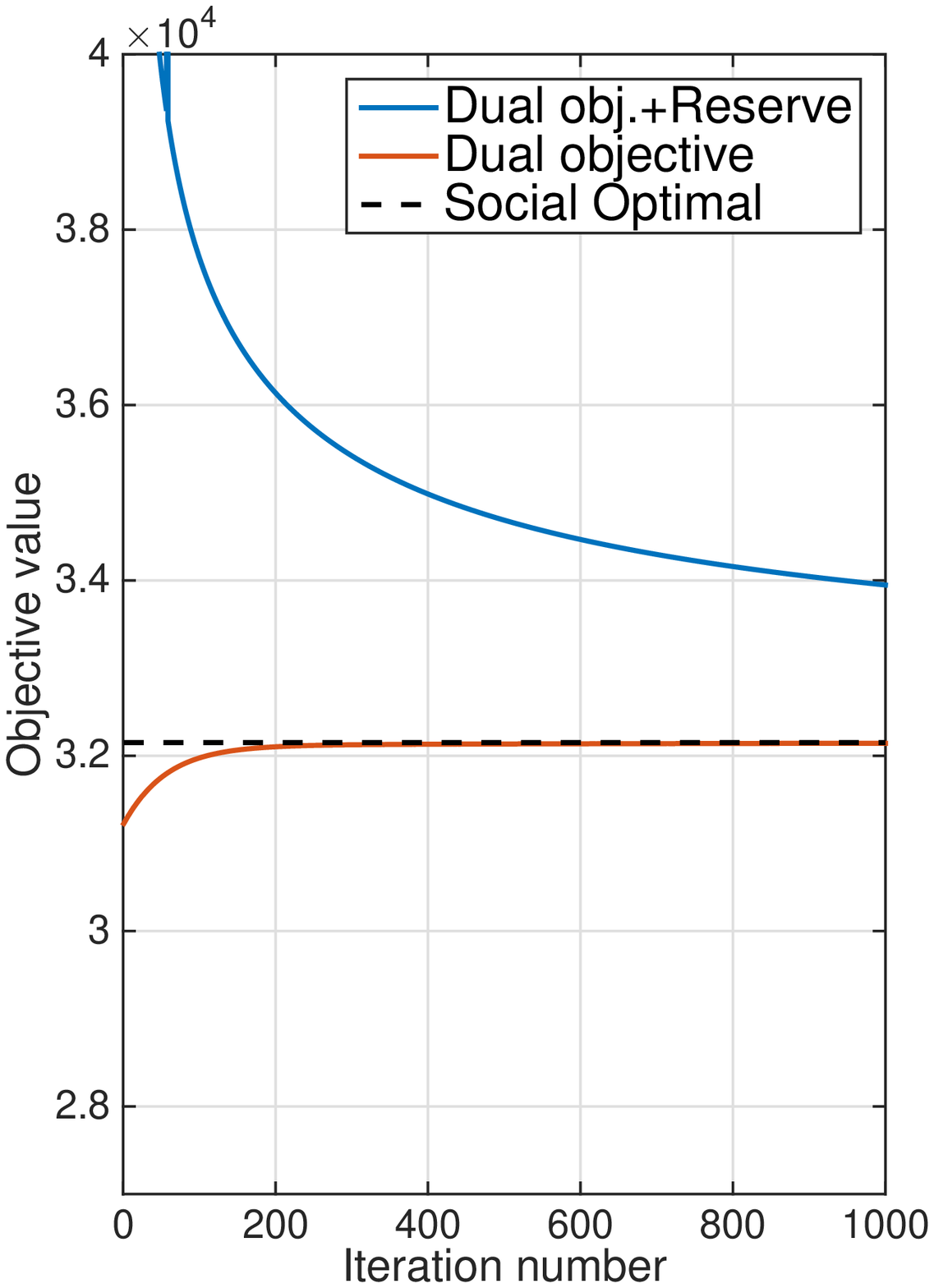}
\caption{(Left) Infeasibility against the iteration number. Notice that primal infeasibility refers to the $\ell_2$-norm $\| (a_k, {\bf w}_k) \|_2$. (Right) Objective value of dual decomposition against the iteration number.} 
\label{fig:infeas} \vspace{-.6cm}
\end{figure}

\section{Conclusions and Future Work}
The implications of large-scale integration of EVs on power and transportation networks, leading to an interdependency between the two infrastructures, were studied under a static setting. We saw that a collaboration between the IPSO and the ITSO can lead EVs towards a socially optimal traffic pattern and energy footprint, and highlighted the adverse effects of ignoring the interconnection between the two infrastructures. We further analyzed the reserve capacity requirements of operating the   grid without a direct collaboration between the IPSO and ITSO. These results were obtained under an ideal static setting and in the absence of retail markets. Important issues remain to be studied in future work. For example, EV charging facilities are expected to be privately-owned, and hence pricing decisions would be left to profit-maximizing entities competing against each other to attract EV drivers to their station. This would affect the IPSO's ability to impose taxes on many arcs in the extended transportation graph and would lower the IPSO's ability to maximize social welfare.  The impact of hourly dynamics of electricity prices and travel patterns is another important aspect that requires further analysis. In this case, non-convexities of the dynamic traffic assignment problem would extend to the IPSO's price design problem.

 {\newblue
\appendices
\section{MPEC for finding $(\gamma^\star, \bm{\mu}^\star)$ in \eqref{eq:primal_infeas}} \label{app:mpec}

To compute the bound \eqref{eq:primal_infeas}, we need an upper bound on  $\| (\gamma^*, \bm{\mu}^* ) \|_2$. To calculate such a bound,  we use an MPEC to enumerate all the possible EV   demand valus  and their corresponding optimal dual variables $(\gamma^*, \bm{\mu}^*)$:
%Without access to the traffic pattern, the IPSO can solve an MPEC to upper bound 
%the optimal dual variables. 
\[
\hspace{-.2cm} \begin{array}{rrl}
%\displaystyle \| (\gamma^0, \bm{\mu}^0) - (\gamma^\star , \bm{\mu}^\star ) \|_2 & \leq \| (\gamma^0, \bm{\mu}^0) \|_2 + \| (\gamma^\star , \bm{\mu}^\star ) \|_2 \\
& \displaystyle  \max_{ {\bf d},\gamma , \bm{\mu}, {\bf z}_L, {\bf z}_U } & \displaystyle  \| (\gamma , \bm{\mu} ) \|_2 \\
& {\rm s.t.} & {\bf 0} \leq {\bf z}_L \leq \delta {\bf 1},~{\bf 0} \leq {\bf z}_U \leq \delta {\bf 1},  \\
& & {\bf d}_{min} \leq {\bf d} \leq {\bf d}_{max}, \vspace{.1cm} \\
& & \min_{\bf g} ~{\bf 1}^T {\bm c}({\bf g}) \\
& & {\rm s.t.} ~~{\bf z}_U :   {\bf g} \leq {\bf g}^{max},~{\bf z}_L : -{\bf g} \leq -{\bf g}^{min}  \\
& & ~~~~~~ \gamma : {\bf 1}^T ({\bf d} + {\bf u} - {\bf g}) = 0 \\
& & ~~~~~~ \bm{\mu} : {\bf H} ({\bf d} + {\bf u} - {\bf g} ) \leq {\bf m},\vspace{-.5cm}
\end{array} \vspace{.5cm}
\]
where ${\bf d}_{min}, {\bf d}_{max}$ are lower/upper bounds to the  electricity demand ${\bf d}$
requested by the EVs and $\delta > 0$ is a regularization parameter for the 
power generation constraints. 
%The optimization finds the maximum possible value of $  \| (\gamma , \bm{\mu} ) \|_2$ when the demand ${\bf d}$ of EVs can take {\it any} value within this range.
As seen in \eqref{lmpdef}, the lower-level minimization problem finds the optimal dispatch 
${\bf g}$ and hence the optimal dual variables $(\gamma , \bm{\mu} )$ 
for {\it each} of the possible demand profiles ${\bf d}_{min} \leq {\bf d} \leq {\bf d}_{max}$. }

\bibliographystyle{IEEEtran}
\bibliography{IEEEabrv,transport,New2,price,science,science2,smart_grid}

\end{document}